\documentclass{llncs}

\usepackage{mathtools,amssymb}
\usepackage{tikz-cd}
\usepackage{multicol}
\usepackage{booktabs}
\usepackage{graphicx}
\usepackage{mathpartir}
\usepackage{array}
\usepackage{url}
\usepackage{algorithm}
\usepackage{algpseudocode}
\usepackage{etoolbox}
\usepackage{wrapfig}
\usepackage{cite}
\usepackage[hidelinks]{hyperref}
\usepackage{thmtools}
\usepackage{thm-restate}

\newcommand{\solvername}{Z3-Noodler\textsubscript{$\mathcal{C}$}}
\newcommand{\Reg}{\mathtt{Reg}}
\newcommand{\SEq}{\mathtt{SEq}}
\newcommand{\LVar}{\mathtt{LVar}}
\newcommand{\Eq}{\mathtt{Eq}}
\newcommand{\Len}{\mathtt{Len}}
\newcommand{\sat}{\mathsf{sat}}
\newcommand{\unsat}{\mathsf{unsat}}
\newcommand{\unknown}{\mathsf{unknown}}
\newcommand{\Var}{\operatorname{Var}}
\newcommand{\eq}{\operatorname{eq}}
\newcommand{\pred}{\operatorname{pred}}
\newsavebox{\mainresultstablebox}
\newcommand{\resultheader}[2]{%
  \shortstack[c]{\strut #1\\\strut (#2)}%
}

\makeatletter
\patchcmd{\ALG@doentity}
  {\ifx\ALG@text\ALG@x@notext\item[]\nointerlineskip\else\item\fi
   \noindent\hskip\ALG@tlm}
  {\ifx\ALG@text\ALG@x@notext\else
     \item\noindent\hskip\ALG@tlm
   \fi}
  {}{\PackageError{main}{Failed to patch algorithmicx hidden rows}{}}
\makeatother

\makeatletter
\patchcmd\WF@putfigmaybe{\lower\intextsep}{}{}{\fail}%
\patchcmd\WF@putfigmaybe
  {\global\advance\WF@size-\intextsep}{}{}{\fail}%
\AddToHook{env/wrapfigure/begin}{\setlength{\intextsep}{0pt}}
\makeatother

\newsavebox{\conditionbox}
\newcommand{\fitcondition}[1]{%
  \begingroup
  \sbox{\conditionbox}{\small$\displaystyle #1$}%
  \ifdim\wd\conditionbox>\textwidth
    \resizebox{\textwidth}{!}{\usebox{\conditionbox}}%
  \else
    \usebox{\conditionbox}%
  \fi
  \endgroup
}

\title{Taming Cyclic Dependencies in Stabilization-Based String Constraint Solving}
\author{Siyi Wu \and Bican Xia}
\institute{
  School of Mathematical Sciences, Peking University, Beijing, China \\
  \email{\{wsyzh@stu.pku.edu.cn,xbc@math.pku.edu.cn\}}
}

\begin{document}
\maketitle

\begin{abstract}
Word equations with cyclic variable dependencies pose a challenge to SMT string
solvers. In stabilization-based procedures, such dependencies can trigger repeated
refinement of regular-language assignments without reaching stability.
We present \solvername{}, an extension of
Z3-Noodler, combining equality-aware stabilization with dedicated loop reasoning.
We introduce equality predicates to retain relevant equations as single stability
obligations, and extend stabilization with equality-aware substitution and
adapted concatenation-eliminating refinement to expose cyclic dependencies
as overlapping variable bindings.
Dedicated loop reasoning rules then exploit the periodic structure
of these bindings, avoiding certain infinite derivations.
We evaluate \solvername{} against Z3-Noodler, cvc5 and OSTRICH2.
On 64 systematically generated, syntactically simple instances
involving cyclic dependencies
and regular constraints, our solver solves all instances, compared with only 20 solved
for the best baseline.
On equation-oriented instances from SMT-LIB'25, \solvername{} achieves the highest coverage, solving 20 more instances than Z3-Noodler while maintaining competitive aggregate CPU time.
\end{abstract}

\section{Introduction}

String constraint solving is a key technology in program analysis
and software verification.  It has been used to verify string-manipulating programs, especially web applications
\cite{bjorner2009path,abdulla2014string,chen2019decision,wassermann2007sound,trinh2014s3,eriksson2023black,loring2019sound,saxena2010symbolic,yu2014automata,li2014symjs},
and analyze cloud-access policies~\cite{backes2018semantic,rungta2022billion}.
Its practical importance has led both to the integration of string reasoning
into general-purpose SMT solvers such as
Z3~\cite{trinh2016progressive,de2008z3,stanford2021symbolic} and
cvc5~\cite{barbosa2022cvc5,liang2016efficient,liang2014dpll}, and to the development of specialized solvers such as
OSTRICH~\cite{hague2025ostrich2,chen2019decision}.

\emph{Stabilization}~\cite{blahoudek2023word} has emerged as an effective
basis for practical string constraint solving.  It reduces satisfiability of
word equations with regular constraints to finding a stable regular assignment
over the variables: the solver
repeatedly refines the assigned regular languages until all language inclusions
induced by the equations become stable. Dependencies among these inclusions are
represented by a directed graph called the inclusion graph, whose vertices are
the inclusions and whose edges capture how a refinement for one predicate may
affect another.
Stabilization forms the basis of the core decision procedure of
Z3-Noodler~\cite{chen2024z3,chen2023solving,havlena2024cooking,
chocholaty2025z3,chocholaty2026string}, in which it is combined with length
reasoning and support for a broader constraint language. Z3-Noodler has
achieved leading performance in the \texttt{QF\_Strings} division at 
SMT-COMP 2024 and 2025~\cite{smtcomp2024strings,smtcomp2025strings}.

Despite this strong empirical performance, stabilization still can fail to
terminate for instances beyond the chain-free~\cite{abdulla2019chain}
setting. An underlying problem is that successive refinement steps
may propagate constraints through a cycle in the inclusion graph and
eventually affect a predicate that has already been stabilized, forcing its
stability to be checked again. Repeated propagation around
the cycle can therefore produce an unbounded sequence of increasingly precise
regular assignments without reaching stability.

This challenge cannot be addressed simply by falling back on a general decision
procedure. Although complete algorithms
for word equations are known~\cite{makanin1986problem,jez2016recompression}, no
implementation of these general algorithms is known to be efficient in
practice. Moreover, the decidability of word equations with linear length
constraints, a combination central to practical SMT solving, has remained open
for years. Therefore, existing solvers typically provide completeness
guarantees for restricted fragments while relying on specialized mechanisms to
handle harder instances. However, even syntactically small instances can remain
challenging when cyclic dependencies are present, as our evaluation demonstrates.

\paragraph{Our approach.}
We present \solvername{}, an SMT solver that extends Z3-Noodler's
stabilization-based decision procedure with equality-aware stabilization and
native loop reasoning for cyclic variable dependencies.  Two key 
observations motivate our work. First, the word-combinatorial relationships
imposed by certain equations involved in cyclic
dependencies can yield substantially stronger constraints on the
candidate solution space than generic inclusion-based refinement, while still
being expressible using ordinary word equations and regular constraints.
Second, the concatenation-elimination stabilization~\cite{chen2023solving},
initially proposed to isolate
variables relevant to length reasoning, can be reused to keep shared
variable occurrences explicit during refinement.

Based on these observations, \solvername{} extends stabilization with equality
predicates to retain each equation involved in cyclic dependencies as a single
stability obligation rather than as two inclusion predicates. This provides a
richer reasoning context than inclusion-level refinement, enabling equation-level
reduction and specialized reasoning that can derive stronger constraints.
At the same time, no equality predicates are introduced for chain-free inputs
by construction. Since the extensions take effect
only in the presence of equality predicates, the procedure coincides with the
original stabilization procedure on the chain-free fragment, preserving its
termination and completeness guarantees. Building on equality predicates, we
develop equality-aware stabilization, which extends
the original procedure with equality-specific
reduction, substitution, and refinement. In particular, we adapt
concatenation-eliminating stabilization to preserve the
boundaries of shared variable occurrences and expose
alignment equations for subsequent substitution.

\paragraph{Results.}
We implement \solvername{} on top of Z3-Noodler 1.6.1 and evaluate it
against Z3-Noodler, cvc5, and OSTRICH2. On a systematically generated diagnostic
benchmark of 64 syntactically small non-chain-free instances combining cyclic
dependencies with regular constraints, 
\solvername{} solves all 64 instances, compared with 12 for Z3-Noodler,
20 for cvc5, and 12 for OSTRICH2.
Disabling the loop reasoning rules in \solvername{} reduces coverage from 64
to 32 instances, providing direct evidence of their contribution. On 25,657
equation-oriented benchmarks from SMT-LIB'25, \solvername{} achieves the highest
coverage while remaining competitive in solving time. 

\section{Preliminaries}
\label{sec:preliminaries}

\paragraph{Regular languages.}
Let $\Sigma$ be a finite alphabet and
$\Sigma_\varepsilon=\Sigma\cup\{\varepsilon\}$.  A \textit{nondeterministic
finite automaton (NFA)} is a tuple $A=(Q,\Sigma,\Delta,Q_I,Q_F)$, where $Q$ is
a finite set of states, $\Delta\subseteq Q\times\Sigma_\varepsilon\times Q$ is
a set of \textit{transitions}, and $Q_I,Q_F\subseteq Q$ are the sets of
initial and final states, respectively.  We write $(q,a,q')\in\Delta$ as
$q\xrightarrow{a}q'$.  A \textit{run} of $A$ over a word $w\in\Sigma^*$ is a
sequence of transitions $q_0\xrightarrow{a_1}q_1,\ldots,
q_{n-1}\xrightarrow{a_n}q_n$ such that $a_1\cdots a_n=w$.  The run is
\textit{accepting} if
$q_0\in Q_I$ and $q_n\in Q_F$.  The \textit{language recognized by} $A$ is
the set $L(A)\subseteq\Sigma^*$ of words over which $A$ has an accepting run.
A language is \textit{regular} if it is recognized by an NFA.  Let
$A'=(Q',\Sigma,\Delta',Q'_I,Q'_F)$ be another NFA with a state set disjoint
from $Q$.  The concatenation $L(A)\cdot L(A')$ is recognized by the NFA
\(A\cdot A'=
  (Q\cup Q',\Sigma,
   \Delta\cup\Delta'\cup
   \{(p,\varepsilon,q)\mid p\in Q_F,\ q\in Q'_I\},Q_I,Q'_F).
\)
The intersection $L(A)\cap L(A')$ is recognized by the NFA
\(
  A\cap A'=(Q\times Q',\Sigma,\Delta^\times,
             Q_I\times Q'_I,Q_F\times Q'_F),
\)
where $(q,q')\xrightarrow{a}(r,r')\in\Delta^\times$ if and only if either
(i)~$a\in\Sigma$, $q\xrightarrow{a}r\in\Delta$, and
$q'\xrightarrow{a}r'\in\Delta'$; or (ii)~$a=\varepsilon$ and either
$q'=r'$ and $q\xrightarrow{\varepsilon}r\in\Delta$, or
$q=r$ and $q'\xrightarrow{\varepsilon}r'\in\Delta'$.

\paragraph{Constraint language.}
Let $\mathbb{X}$ be a set of string variables and $\mathbb{I}$ a set of
integer variables.  We use $x_S\in\mathbb{X}$ for
string variables, $x_I\in\mathbb{I}$ for integer variables, $c_I\in\mathbb{Z}$ for an integer
constant, and $a\in\Sigma$ for an alphabet symbol.  The abstract syntax of
the constraint language is given by the following grammar:
\[
\begin{array}{rcl}
  e_S &::=& \varepsilon \mid x_S \mid e_Se_S,\\
  e_I &::=& c_I \mid x_I \mid \mathrm{len}_{x_S} \mid (e_I)
    \mid e_I+e_I \mid c_I\cdot e_I,\\
  \alpha_I &::=& e_I\leq e_I,\\
  L &::=& \varnothing \mid \{\varepsilon\} \mid \{a\}
    \mid L\cup L \mid L\cdot L \mid L^*,\\
  \alpha_S &::=& e_S\approx e_S \mid x_S\in L,\\
  \varphi &::=& \top \mid \bot \mid \alpha_I \mid \alpha_S
    \mid (\varphi) \mid \neg\varphi
    \mid \varphi\land\varphi
    \mid \varphi\lor\varphi.
\end{array}
\]
The expressions generated by $L$ denote exactly
the regular languages over $\Sigma$, where $\cup$, $\cdot$, and ${}^*$ denote union,
concatenation and Kleene star, respectively.
For a string term $e_S$, its \emph{term size} $|e_S|$ is the number of variable
occurrences in $e_S$.
We treat a word $w$ as syntactic sugar for a fresh string variable
$x_S$ together with the singleton-language constraint $x_S\in\{w\}$.  A
\emph{string constraint} is either a \emph{word equation} of the form
$e_S\approx e_S$ or a \emph{regular constraint} of the form $x_S\in L$, where
$L$ is a regular language.
An equation is called a
\emph{binding} if it has the form $x_S\approx e_S$ or
$e_S\approx x_S$.  For a formula $\varphi$, the set
$\Var(\varphi)\subseteq\mathbb{X}$ denotes the set of all string variables in
$\varphi$.
For formulas $\varphi$ and $\psi$, let $Z$ be the set of all variables, of
either sort, that occur in $\varphi$ but not in $\psi$.  We call $\varphi$ an
\emph{existential extension} of $\psi$ if
$\psi\Longleftrightarrow\exists Z\,\varphi$.

\paragraph{Normalized positive cubes.}
We call a conjunction of atoms a
\emph{positive cube}.  Without loss of generality, we identify every positive
cube $\varphi$ with an equivalent normalized form
$\Eq_\varphi\land\Reg_\varphi\land\Len_\varphi$.
Here, $\Eq_\varphi$ contains all equations in $\varphi$ except trivial
identities, and
$\Len_\varphi$ contains all length constraints in $\varphi$.  The regular part
$\Reg_\varphi$ contains exactly one atom $x_S\in L_{x_S}$ for every
$x_S\in\Var(\Eq_\varphi)$.  It is obtained by intersecting all regular
constraints on the same variable and assigning $\Sigma^*$ to every variable
in $\Var(\Eq_\varphi)$ for which no regular constraint is given.  We identify
each conjunction with the set of its
conjuncts and omit the subscript $\varphi$ when it is clear from the context.

\paragraph{Regular assignments.}
We regard the normalized regular part $\Reg_\varphi$ as a
\emph{regular-language assignment}, or \emph{regular assignment} for short,
and write $\Reg_\varphi(x)=L_x$ whenever $x\in L_x$ occurs in
$\Reg_\varphi$.  We define a \emph{string assignment} over a set
$X\subseteq\mathbb{X}$ to be a function $\nu\colon X\to\Sigma^*$.
We view a regular assignment $\Reg$ with domain $X$ as representing the set
$[\![\Reg]\!]\coloneqq
\{\nu\mid\forall x\in X,\nu(x)\in\Reg(x)\}$
of candidate string assignments.  We call $\Reg$ \emph{feasible} if
$[\![\Reg]\!]\neq\varnothing$; equivalently,
$\Reg(x)\neq\varnothing$ for every $x$ in its domain. 
We extend $\nu$ and $\Reg$ homomorphically to string terms, with
$\nu(\varepsilon)=\varepsilon$, $\Reg(\varepsilon)=\{\varepsilon\}$,
$\nu(st)=\nu(s)\nu(t)$, and $\Reg(st)=\Reg(s)\cdot\Reg(t)$.
For every $\nu\in[\![\Reg]\!]$ and every string term $t$ over $X$, we have
$\nu(t)\in\Reg(t)$. A regular assignment $\Reg'$ \emph{refines} $\Reg$ if they have the same
domain $X$ and $\Reg'(x)\subseteq\Reg(x)$ for every $x\in X$.  It follows
that $[\![\Reg']\!]\subseteq[\![\Reg]\!]$; thus, refinement can only
eliminate candidate string assignments.

\paragraph{Stability characterization.}
Given string terms \(s\) and \(t\) and a regular assignment $\Reg$, we write
\(s\supseteq_{\Reg} t\) if \(\Reg(s)\supseteq\Reg(t)\), and
\(s=_\Reg t\) if \(\Reg(s)=\Reg(t)\).
We call these an \emph{inclusion predicate} and an \emph{equality predicate},
respectively; predicates of either form are called \emph{stability predicates}.
We write \(p(s,t)\) for a stability predicate \(p\) with sides \(s\) and \(t\),
and define its underlying word equation by \(\eq(p)\coloneqq s\approx t\).
We extend \(\eq\) to a set \(P\) of stability predicates by defining
\(\eq(P)\coloneqq \{\eq(p)\mid p\in P\}\).
Conversely, for a word equation \(e=(s\approx t)\), we define the corresponding
equality predicate by \(\pred_{\Reg}(e)\coloneqq (s=_\Reg t)\), and extend \(\pred_{\Reg}\) to a set \(\Eq\) of word equations by defining
\(\pred_{\Reg}(\Eq)\coloneqq\{\pred_{\Reg}(e)\mid e\in \Eq\}\). We omit the subscript \(\Reg\) when it is clear from the context.
For a system of equations \(\Eq\), the procedure \(\mathtt{incl}\) from
\cite{blahoudek2023word} constructs a set $V$ of inclusion predicates that
contains, for each equation \(s\approx t\in\Eq\), at least one of the two
inclusions \(s\supseteq t\) and \(t\supseteq s\).
Define the dependency relation
\(
E=\{(s\supseteq t,s'\supseteq t')\in
  V\times V\mid
  \Var(t)\cap\Var(s')\ne\varnothing\}
\).
The directed graph \((V,E)\) is called the \emph{inclusion graph} of \(\Eq\).  A directed
cycle in this graph represents a \emph{cyclic dependency}. We say that \(\Eq\) is
\emph{chain-free} if \(\mathtt{incl}(\Eq)\) is acyclic, or equivalently, if $V$
contain no pair of opposite inclusions.
We say that a regular assignment \(\Reg\) is \emph{stable} for \(V\)
if every predicate in \(V\) holds with respect to \(\Reg\).  By \cite[Theorem~3]{blahoudek2023word},
\(\Eq\land\Reg\) is satisfiable iff there exists a feasible refinement
\(\Reg'\) of \(\Reg\) that is stable for \(V\).

\section{Motivating Example}
\label{sec:motivation}

Consider the positive cube $xy\approx yx\land x\in\texttt{.*a.*}\land y\in\texttt{b+}$, in which $x\in\texttt{.*a.*}$ means that $x$ must
contain at least one occurrence of the character \texttt{a}, and $y\in\texttt{b+}$ means that $y$ is a
nonempty sequence of the character \texttt{b}.
It is unsatisfiable: $xy\approx yx$ implies that $x$ and $y$ are
powers of a common string $z$.  Since $y\in\texttt{b+}$, neither $z$ nor
$x$ can contain \texttt{a}, contradicting $x\in\texttt{.*a.*}$.
Despite this short proof of unsatisfiability, the instance is challenging for
several prominent string solvers including Z3-Noodler 1.6.1,
cvc5 1.3.4, and OSTRICH 2.1.  Under a 120-second CPU-time limit and an 8~GB
memory limit, each solver either returns $\unknown$ or times out on this
instance.

By the stability characterization, 
Z3-Noodler~\cite{chen2024z3} decides this instance by searching for a
feasible refinement $\Reg'$ of $\Reg$ such that
$\Reg'(xy)=\Reg'(yx)$, or by showing that no such refinement exists.  For this instance, the stabilization procedure works with the two directional
inclusion predicates $yx\supseteq xy$ and $xy\supseteq yx$.
Each state carries a regular assignment $\Reg$ and a FIFO worklist
$\mathsf{F}$ of inclusion predicates whose stability remains to be checked.
When $\mathsf{F}$ is nonempty, the procedure removes its first predicate
$s\supseteq t$ and checks it under the current regular assignment $\Reg$.
If the predicate is satisfied, the procedure proceeds to the next one;
Otherwise, the procedure applies the refinement construction of
\cite{blahoudek2023word}: it intersects $\Reg(t)$ with $\Reg(s)$ and
enumerates the feasible factorizations of the intersection along the variable
occurrences of $t$, yielding finitely many refinements
$\{\Reg_i\}_{i=1}^k$ of $\Reg$.
Each $\Reg_i$ is then explored as separate branches. Semantically, these refinements restrict the candidate string assignments
to $C_\supseteq\coloneqq\{\nu\in[\![\Reg]\!]\mid \nu(t)\in\Reg(s)\}$,
in the sense that $C_\supseteq=\bigcup_{i=1}^{k}[\![\Reg_i]\!]$.

Crucially, $C_\supseteq$ is defined with respect to the pre-refinement language $\Reg(s)$.
When $s$ and $t$ share variables, restricting the
variables occurring in $t$ may simultaneously shrink $\Reg_i(s)$. Hence,
although the refinement ensures $\Reg_i(t)\subseteq\Reg(s)$, it does not
necessarily ensure $\Reg_i(t)\subseteq\Reg_i(s)$. The predicate
$s\supseteq t$ may therefore become unstable again and must
re-enter the worklist.
Table~\ref{tab:motivating-refinements}  traces how this behavior unfolds in
the motivating example.

\begin{table}[h]
  \centering
  \setlength{\tabcolsep}{5pt}
  \renewcommand{\arraystretch}{1.3}
  \caption{Successive refinement states generated by the stabilization procedure of Z3-Noodler~\cite{chen2024z3}.}
  \label{tab:motivating-refinements}
  \resizebox{.9\textwidth}{!}{%
    \begin{tabular}{cccccc}
      \toprule
      \# & \multicolumn{1}{c}{$\Reg(x)$} & \multicolumn{1}{c}{$\Reg(y)$} & Predicates to process & \multicolumn{1}{c}{Target constraint} \\
      \midrule
      0 & \texttt{.*a.*} & \texttt{b+}
        & ${yx \supseteq xy},\; xy \supseteq yx$ & --- \\
      1 & \texttt{b+.*a.*} & \texttt{b+}
        & ${xy \supseteq yx},\; yx \supseteq xy$
        & $\Reg(xy) \subseteq \text{\texttt{b+.*a.*}}$ \\
      2 & \texttt{b+.*a.*b+} & \texttt{b+}
        & ${yx \supseteq xy},\; xy \supseteq yx$
        & $\Reg(yx) \subseteq \text{\texttt{b+.*a.*b+}}$ \\
      3 & \texttt{bb+.*a.*b+} & \texttt{b+}
        & ${xy \supseteq yx},\; yx \supseteq xy$
        & $\Reg(xy) \subseteq \text{\texttt{bb+.*a.*b+}}$ \\
      4 & \texttt{bb+.*a.*b+b} & \texttt{b+}
        & ${yx \supseteq xy},\; xy \supseteq yx$
        & $\Reg(yx) \subseteq \text{\texttt{bb+.*a.*b+b}}$ \\
      $\hdots$ & $\hdots$ & $\hdots$ & $\hdots$ & \\
      \bottomrule
    \end{tabular}%
  }
\end{table}

Initially, the procedure processes $yx\supseteq xy$.   It refines $xy$ against
the current language $\Reg(yx)=\texttt{b+.*a.*}$, shrinking $\Reg(x)$ to
$\texttt{b+.*a.*}$ in state~1. 
Since $x$ also occurs in $yx$, this refinement
simultaneously shrinks $\Reg(yx)$. Consequently, the predicate $yx\supseteq xy$ is still not guaranteed to hold and is
appended back to the worklist. The procedure then processes $xy\supseteq yx$ analogously,
producing state~2 and
invalidating $xy\supseteq yx$ in turn.
As shown in the table, the states exhibit a two-step recurrence: after every
two refinements, the worklist returns to the same predicate order, while the
minimum number of leading and trailing $\texttt{b}$s required by $\Reg(x)$
increases by one.
The two predicates are thus repeatedly refined and re-enqueued, generating an
infinite sequence of states without reaching a stable assignment
or exhausting the search. 

The repeated refinements above reveal a limitation of the generic refinement
rule: it reasons only about the current regular languages, without exploiting
the additional structure imposed by cyclic variable dependencies. For example,
when processing $yx\supseteq xy$, the set of solutions satisfying the underlying
equation $xy\approx yx$ is $C^\star\coloneqq
\{\nu\in[\![\Reg]\!]\mid \nu(xy)=\nu(yx)\}$, while the generic refinement
rule only restricts the candidate solution space to $C_{\supseteq}\coloneqq
\{\nu\in[\![\Reg]\!]\mid \nu(xy)\in\Reg(yx)\}$.
Our loop reasoning instead exploits the equation structure to obtain an \emph{exact}
decomposition while remaining expressible using only word equations and regular
constraints supported by the existing constraint language. 
Since $y\in\texttt{b+}$ implies $y\not\approx\varepsilon$,
Lemma~\ref{lem:loop-reasoning-soundness} in Sect.\ref{sec:cyclic-reasoning} shows that
$xy\approx yx\land \Reg$ is equivalent to $\exists x_1\ \exists x_2\ \exists z\ (\psi\land\Reg')$, where
\[
\psi\coloneqq\bigl(x\approx x_1x_2\land y\approx x_2z\land x_2z\approx zx_2),\quad \Reg'\coloneqq\Reg\land x_1\in\Reg(y)^*\land z\in\Sigma^+\bigr.
\]
Applying this transformation to the motivating example yields state~1 in
Table~\ref{tab:example-fixed}, where $x\approx x_1x_2$ in the decomposition is
represented by the one-sided stability obligation $x\supseteq x_1x_2$.
The resulting regular constraints expose the conflict after two ordinary
stabilization steps.

\begin{table}[h]
  \centering
  \caption{Stabilization-based proof of unsatisfiability for the motivating example.}
  \label{tab:example-fixed}
  \setlength{\tabcolsep}{5pt}
  \renewcommand{\arraystretch}{1.3}
  \resizebox{.9\textwidth}{!}{%
    \begin{tabular}{c
      >{\raggedright\arraybackslash}m{.3\linewidth}
      >{\raggedright\arraybackslash}m{.26\linewidth}
      c}
      \toprule
      \# & \multicolumn{1}{c}{$\Reg$} & Predicates to process & \multicolumn{1}{c}{Target constraint} \\
      \midrule
      0 & $x\in \texttt{.*a.*},\;y\in \texttt{b+}$
        & ${xy = yx}$ & --- \\
      1 & $x\in \texttt{.*a.*},\;y\in \texttt{b+},\allowbreak x_1\in\texttt{b*},x_2\in\texttt{.*},\;z\in\texttt{.+}$
        & ${x\supseteq x_1x_2}, y= x_2z,\allowbreak x_2z= zx_2, x_1y= yx_1$
        & --- \\
      2 & $y\in \texttt{b+},\;x_1\in\texttt{b*},\allowbreak x_2\in\texttt{.*a.*},\;z\in\texttt{.+}$
        & ${y= x_2z},x_2z= zx_2,\allowbreak x_1y= yx_1$
        & $\Reg(x_1x_2) \subseteq \texttt{.*a.*}$ \\
      3 & $x_1\in\texttt{.*},\;x_2\in\varnothing,\allowbreak z\in\texttt{b+}$
        & $x_2z= zx_2,x_1y= yx_1$
        & $\Reg(x_2z) \subseteq \texttt{b+}$ \\
      \bottomrule
    \end{tabular}%
  }
\end{table}

The example above shows how cyclic variable dependencies can lead to repeated
refinement in the generic stabilization procedure, and how this behavior can
be mitigated through specialized reasoning. To assess whether this phenomenon
extends beyond the motivating example, we constructed a diagnostic benchmark
of 64 analogous instances by combining 16 equation templates with 4 regular-constraint
configurations. Under the experimental setup described in
Sect.~\ref{sec:evaluation}, Z3-Noodler, cvc5, and OSTRICH left 52, 44, and 52
of these instances unresolved, respectively. These results indicate that the
difficulty exposed by the motivating example is not isolated.
Sect.~\ref{sec:evaluation} describes the benchmark construction in detail and
also evaluates our solver \solvername{} on it.

\section{Architecture}
\label{sec:architecture}
\paragraph{Overview.}
Inspired by the analysis in the previous section, we implement our solver
\solvername{} on
top of Z3-Noodler 1.6.1~\cite{chocholaty2025z3}. \solvername{} adapts the core string
decision procedure in two ways. First, it extends the
existing stabilization framework with equality predicates while preserving the
handling of inclusion predicates, and extends the stabilization procedure with
equality-specific handling (Sect.~\ref{sec:stabilization}).
Second, it incorporates loop reasoning rules, illustrating a way to complement
the generic stabilization procedure with specialized reasoning capabilities
(Sect.~\ref{sec:cyclic-reasoning}).

Overall, our extension leaves the
surrounding solving pipeline of Z3-Noodler unchanged. Given an SMT-LIB input
\cite{barrett2024smtlib}, \solvername{} reuses Z3-Noodler's existing front-end
to prepare the input for SMT solving. Z3's SAT solver
then iteratively provides the string theory solver with conjunctions of string
literals, each corresponding to an assignment of the Boolean abstraction of the input
formula. To solve each conjunction, \solvername{} first considers the
applicable specialized decision procedures of Z3-Noodler; if none yields a
conclusive result, the conjunction is passed to the stabilization-based decision
procedure. \solvername{} also preserves Z3-Noodler's existing support for various
extended string constraints, such as string--integer conversions~\cite{havlena2024cooking}, negated
containment~\cite{chen2025uniform} and transducer constraints~\cite{chocholaty2026string}. Our extension modifies only the main and the most general 
stabilization procedure, which takes as input a positive cube in the normalized form
defined in Sect.~\ref{sec:preliminaries}.\footnote{The extension retains
Z3-Noodler's existing support for transducer constraints
but does not broaden it to the non-chain-free setting. For clarity, we limit our discussion to transducer-free inputs.} Finally, the outcome of the
stabilization procedure is sent back to the SAT solver, which continues the
global SMT search.

\subsection{Search procedure}
\label{sec:search-procedure}
\begin{wrapfigure}{R}{0.52\textwidth}
  \begin{minipage}[t]{\linewidth}
    \begin{algorithm}[H]
      \caption{Stabilization-based decision procedure adapted from~\cite{chocholaty2026string}}
      \label{alg:search-procedure}
      \algrenewcommand{\algorithmicindent}{0.8em}
      \algrenewcommand{\algorithmicthen}{}
      \begin{algorithmic}[1]
        \algtext*{EndIf}
        \algtext*{EndFor}
        \algtext*{EndWhile}
        \Statex \textbf{Input:} A positive cube $\varphi$
        \Statex \textbf{Output:} $\sat/\unsat$
        \State $\mathit{rules} \leftarrow [\textsc{OvWord},\textsc{OvPre},\textsc{OvSuf},$
        \Statex \hspace{\algorithmicindent}$\textsc{Skip},\textsc{EqSkip},\textsc{Subst},\textsc{CombNL},\textsc{Refine}]$
        \State $V \leftarrow$ the vertex set of $\mathtt{incl}(\Eq_\varphi)$\label{line:search-incl}
        \State $\mathsf{N} \leftarrow \mathtt{toposort}(\mathtt{merge}(V))$\label{line:search-merge}
        \State $W \leftarrow \langle(\mathsf{N},\mathsf{N},\Reg_{\varphi},\varnothing,\Var(\Len_{\varphi}))\rangle$\label{line:search-root}
        \While{$W\ne\varnothing$}
          \State $S \leftarrow W.\mathtt{dequeue}()$\label{line:search-pop}
          \State $S \leftarrow \texttt{ReduceEq}(S)$\label{line:search-reduce}
          \If{$\Reg_S$ is $\unsat$}\label{line:search-reg-unsat}
            \State \textbf{continue}\label{line:search-reg-prune}
          \EndIf
          \If{$\mathsf{F}_S=\varnothing$}\label{line:search-terminal}
            \If{$\texttt{LenImg}(\Reg_S\land\SEq_S)\land\Len_\varphi$ is $\sat$}\label{line:search-length}
              \State \Return $\sat$\label{line:search-sat}
            \EndIf
            \State \textbf{continue}\label{line:search-continue}
          \EndIf
          \For{$\mathit{rule} \in \mathit{rules}$}\label{line:search-rules}
            \If{$\mathit{rule}$ is applicable to $S$}
              \State $W.\mathtt{enqueue}(\mathit{rule}(S))$
              \State \textbf{break}\label{line:search-break}
            \EndIf
          \EndFor
        \EndWhile
        \State \Return $\unsat$\label{line:search-unsat}
      \end{algorithmic}
    \end{algorithm}
  \end{minipage}
\end{wrapfigure}
Algorithm~\ref{alg:search-procedure} outlines the main decision procedure of \solvername{}.
\solvername{} performs a breadth-first search in which each state $S$ is represented as a 
tuple $(\mathsf{N}_S,\mathsf{F}_S,\allowbreak\Reg_S,\SEq_S,\LVar_S)$.  Here, $\mathsf{N}_S$ is the set 
of stability predicates that must be satisfied, and $\mathsf{F}_S$ is 
the processing frontier,
consisting of the predicates in $\mathsf{N}_S$ whose stability has not yet been established.
The processing frontier always respects the dependency order of predicates,
ensuring that all inclusions precede equalities and the inclusions follow a topological 
order of the inclusion graph. Consequently, whenever the head of the frontier is an inclusion predicate, it does not depend on any other inclusion predicate remaining in the frontier.
$\Reg_S$ is the current regular assignment, normalized over $\Var(\mathsf{N}_S)$.
$\SEq_S$ records the variable bindings introduced by substitutions, tracking
how the current variables relate to the input variables.  Finally, $\LVar_S$ is the set of
length-aware variables. We write $\Phi(S)\coloneqq\eq(\mathsf{N})\land \Reg \land \SEq$
for the formula represented by $S$.

On input $\varphi$, the procedure first obtains the vertex set $V$ of the
inclusion graph of $\Eq_\varphi$ (line~\ref{line:search-incl}).  Then the procedure merges
each pair of opposite inclusions into an equality predicate and sorts the
remaining inclusions in the topological order (line~\ref{line:search-merge}). 
The search starts from the state $(\mathsf{N},\mathsf{N},\Reg_\varphi,\varnothing,\Var
(\Len_\varphi))$ in which $\mathsf{N}$ forms both the initial stability requirements and 
the initial processing frontier, the regular assignment is initialized to $\Reg_\varphi$,
no substitution bindings have yet been introduced, and the variables occurring in
the input length constraints are marked as length-aware (line~\ref{line:search-root}).
The procedure iteratively dequeues the front state $S$ from $W$~(line~\ref{line:search-pop}).
After a state $S$ is chosen, the procedure first invokes \texttt{ReduceEq} as a
proprocessing step to reduce equality predicates~(line~\ref{line:search-reduce}).
If $\Reg_S$ is infeasible, the state is
immediately pruned (lines~\ref{line:search-reg-unsat}--\ref{line:search-reg-prune}).
Otherwise, an empty frontier $\mathsf{F}_S$ indicates that every stability predicate in
$\mathsf{N}_S$ holds with respect to $\Reg_S$, and that variables in $\Var(\SEq_S)\cup\LVar_S$ 
may occur in $\mathsf N_S$ only on the right-hand sides of inclusion predicates.
For such a state, 
the procedure invokes $\texttt{LenImg}$ from \cite{chocholaty2026string} to construct
a quantifier-free linear arithmetic formula $\texttt{LenImg}(\Reg_S\land\SEq_S)$ that, when conjoined with
$\Len_\varphi$, is equisatisfiable with
$\Phi(S)\land\Len_\varphi$. If
$\texttt{LenImg}(\Reg_S\land\SEq_S)\land\Len_\varphi$ is satisfiable, the procedure
returns $\sat$ (lines~\ref{line:search-length}--\ref{line:search-sat}). Otherwise, the
state $S$ is pruned and the search continues with the next state
in $W$ (line~\ref{line:search-continue}).  If
$\mathsf{F}_S\ne \varnothing$, the procedure scans the rules in their prescribed
order and applies the first rule applicable to $S$, adding the successor states to $W$
(lines~\ref{line:search-rules}--\ref{line:search-break}).  Finally, if the worklist $W$
is exhausted without reaching a satisfiable state, it returns $\unsat$ (line~\ref{line:search-unsat}).

\subsection{Equality-Aware Stabilization}
\label{sec:stabilization}
\paragraph{Core inference rules.}
Figure~\ref{fig:adapted-rules} presents the core stabilization rules used in
Algorithm~\ref{alg:search-procedure}. The presentation follows
\cite{chocholaty2026string}, but is simplified to omit machinery not needed in
our setting. A rule written as $\textsc{Rule}:\dfrac{\mathcal{S}}{\{S_1,\ldots,S_k\}}\;\varphi$
applies to a state $S$ matching the \textit{source pattern} $\mathcal{S}$ when
the \textit{side condition} $\varphi$ holds, and produces
the \textit{successors} $\textsc{Rule}(S)\coloneqq\{S_1,\ldots,S_k\}$.  In
particular, the frontier component in every source pattern here has the form
$p::\mathsf{F}$: the rule processes the stability predicate $p$, and the
remaining frontier predicates are bound to $\mathsf{F}$
in the side conditions and conclusions. When updating the frontier, the union operator
$\cup$ is understood as an ordered insertion:
newly added inclusion predicates are prepended
in topological order, while equality predicates
are appended to the end.  Duplicate predicates
are not inserted.

Overall, \solvername{} inherits
the rules \textsc{Skip} and \textsc{CombNL}, introduces a new
rule \textsc{EqSkip}, and generalizes \textsc{Subst} and \textsc{Refine} 
to equality predicates. The equality-specific cases are the only
changes to the rules, and equality
predicates can arise only from a pair of inclusions merged at the root or
from processing a state that already contains an equality predicate. For chain-free
inputs, the initial state contains only inclusion predicates, so the rules behave
exactly as in the original procedure.

\begin{figure}[tb]
  \centering
  {\setlength{\fboxsep}{3pt}%
  \fbox{%
  \begin{minipage}{0.96\textwidth}
  \setlength{\textwidth}{\linewidth}
  \setlength{\abovedisplayskip}{2pt}
  \setlength{\belowdisplayskip}{2pt}
  \centering
  \fitcondition{
      \textsc{Skip}\;
      \dfrac{(\mathsf{N},s\supseteq t::\mathsf{F},\Reg,\SEq, \LVar)}
      {(\mathsf{N},\mathsf{F},\Reg,\SEq,\LVar)}
      \quad \Var(s)\cap\LVar=\varnothing,\; s\supseteq_{\Reg}t
  }\par\vspace{1.2ex}

  \fitcondition{
      \textsc{EqSkip}\;
      \dfrac{(\mathsf{N},s=t::\mathsf{F},\Reg,\SEq, \LVar)}
      {(\mathsf{N},\mathsf{F},\Reg,\SEq,\LVar)}
      \quad \Var(st)\cap\LVar=\varnothing,\; s=_{\Reg}t
  }\par\vspace{1.2ex}

  \fitcondition{
      \textsc{Subst}\;
      \dfrac{(\mathsf{N}\uplus\{p\},p::\mathsf{F},\Reg,\SEq,\LVar)}
      {(\sigma(\mathsf{N}),\sigma(\mathsf{F}\cup\mathsf{F}'),\Reg,
        \sigma(\SEq)\cup\{x\approx s\},\LVar)}
      \quad\varphi_{\textsc{Subst}}
  }\par\vspace{-0.2ex}
  \fitcondition{
    \begin{aligned}
      \varphi_{\textsc{Subst}}\overset{\mathrm{def}}{\Longleftrightarrow}\;
        & eq(p)\text{ is a binding }x\approx s\text{ or }s\approx x,\quad x\notin\Var(s),\\
        & \sigma=\{x\mapsto s\},\quad x\in\LVar\Rightarrow\Var(s)\subseteq\LVar,\quad x\supseteq_\Reg s,\\
        & \mathsf{F}'=\{s\supseteq t\in\mathsf{N}\mid x\in\Var(s) \} \cup \{s=t\in\mathsf{N}\mid x\in\Var(st)\}
    \end{aligned}
  }\par\vspace{1.2ex}

  \fitcondition{
      \textsc{CombNL}\;
      \dfrac{(\mathsf{N}\uplus\{uvw\supseteq t\},uvw\supseteq t::\mathsf{F},\Reg,\SEq, \LVar) }
      {(\mathsf{N}\cup\{uxw\supseteq t,v\supseteq x\},\mathsf{F}\cup\{uxw\supseteq t\},
        \Reg\cup\{x\mapsto\Reg(v)\},\SEq,\LVar)}
      \quad\varphi_{\textsc{CombNL}}
  }\par\vspace{-0.2ex}
  \fitcondition{
      \varphi_{\textsc{CombNL}}\overset{\mathrm{def}}{\Longleftrightarrow}
      |v|>1,\quad
      \Var(v)\cap\LVar=\varnothing
  }\par\vspace{1.2ex}

  \fitcondition{
      \textsc{Refine}\;
      \dfrac{(\mathsf{N}\uplus\{p\},p::\mathsf{F},\Reg,\SEq, \LVar)}
      {\{(\mathsf{N}\cup\texttt{in}_i\cup\texttt{out}_i,
          \mathsf{F}\cup\texttt{in}_i\cup\texttt{out}_i,
          \Reg_i,\SEq,\LVar\cup\LVar_i)\}_{i=1}^k}
      \quad\varphi_{\textsc{Refine}}
  }\par\vspace{-0.2ex}
  \fitcondition{
      \begin{aligned}
      \varphi_{\textsc{Refine}}\overset{\mathrm{def}}{\Longleftrightarrow}
        {}&\texttt{Noodlify}(p, \Reg, \LVar) 
        =\{(\Reg_i,\texttt{in}_i,\texttt{out}_i)\}_{i=1}^k,\\[-0.4ex]
      &\LVar_i=
        \bigcup_{\substack{x\in\LVar\\q\in\texttt{in}_i\cup\texttt{out}_i\\}}
        \{y\mid x,y\in\Var(q)\}
      \end{aligned}
  }\par
  \end{minipage}%
  }}
  \caption{The core inference rules of Algorithm~\ref{alg:search-procedure}.}
  \label{fig:adapted-rules}
\end{figure}

Rules \textsc{Skip} and \textsc{EqSkip} discharge stability predicates that
already hold under the current regular assignment when no further
concatenation elimination is required for length reasoning.  The inherited
\textsc{Skip} rule only requires the left-hand side $s$ to contain no
length-aware variable, exactly as in the original procedure.  Equality
predicates are symmetric, so \textsc{EqSkip} requires
$\Var(st)\cap\LVar=\varnothing$.

Rule \textsc{Subst} consumes a binding whose underlying equation is
$x\approx s$ (up to symmetry) and records the exact binding in $\SEq$ while
substituting $s$ for $x$ in the remaining stability obligations.  The
condition $x\notin\Var(s)$ ensures that
$\sigma=\{x\mapsto s\}$ is a proper elimination substitution.  The condition
$x\supseteq_{\Reg}s$ guarantees that replacing $x$ by $s$ is compatible with
the current regular restriction on $x$.  If $x$ is length-aware, all variables
introduced in its place must remain length-aware, which explains the condition $x\in\LVar\Rightarrow\Var(s)\subseteq\LVar$. The reactivation set $\mathsf{F}'$
is asymmetric for inclusions and symmetric
for equalities.  For each $(u\supseteq v)\in\mathsf{N}$, replacing $x$ by the more
precise term $s$ can invalidate this inclusion only when
$x$ occurs in the left-hand side $u$.  Therefore, only
such inclusions are returned to the frontier.  For an equality $u=v$, a change
on either side can destroy language equality, so every equality containing
$x$ is returned to the frontier.

The rule order in Algorithm~\ref{alg:search-procedure} is important here.
For a one-way inclusion binding without length-aware variables,
\textsc{Skip} is applied as soon as that inclusion is stable, exactly as in
the original procedure.  For an equality binding, however,
\textsc{EqSkip} is applicable only after both directions are stable.  Hence,
when $x\supseteq_{\Reg}s$ but $s\not\supseteq_{\Reg}x$,
\textsc{Subst} can eliminate $x$ while the exact equation
$x\approx s$ is retained in $\SEq$.  This is the equality-aware substitution
step that allows dependencies distributed over several equations to be
propagated by successive substitutions.

\paragraph{Equality-aware refinement.}
Rule \textsc{CombNL} is unchanged and applies only to inclusion
predicates; hence, it never merges variable occurrences in an
equality predicate.
Rule \textsc{Refine} invokes \texttt{Noodlify}
to refine the current regular assignment $\Reg$ with respect to
a selected stability predicate $p$. For inclusion predicates,
\texttt{Noodlify} follows the original refinement procedure
of~\cite{blahoudek2023word}. For a equality predicate
$p\coloneqq(s=t)$, it instead applies the concatenation-eliminating
noodlification of~\cite{chen2023solving}. Since variable occurrences
on both sides remain separate,
each occurrence receives its own binding equation, with fresh
auxiliary variables representing common fragments identified by
the alignment. The resulting binding equations are retained as
explicit equality predicates in
$\texttt{in}_i\cup\texttt{out}_i$, allowing subsequent applications
of \textsc{Subst} to propagate these bindings. Unlike the
original inclusion-based treatment, which may refine
$s\supseteq t$ and $t\supseteq s$ independently when both
directions are present, the equality-aware procedure constructs
a joint alignment for both directions, avoiding redundant
decompositions. Moreover, occurrences
of distinct variables may overlap in an alignment, and fresh
auxiliary variables make their common fragments explicit.
Substituting the resulting bindings into other equations can
transform dependencies among distinct variables into a local
self-overlap of an auxiliary variable recognizable
by local reasoning rules.

\begin{example}
\label{ex:expose}
Suppose concatenation-eliminating noodlification of an equality introduces
two bindings $z=x$ and $xu=y$, where $x$ is a fresh auxiliary variable. If
the current state also contains an equality $sy=zt$, substituting these
bindings yields $sxu=xt$. Thus, the overlap between the distinct variables
$z$ and $y$, represented by $x$, becomes a self-overlap of $x$ after
substitution. Longer dependency
cycles can be handled analogously by successive refinement and substitutions.
\end{example}

\paragraph{Equality reduction.}
Stabilization on non-chain-free inputs may generate new equality predicates whose two
sides share variables, creating new reduction opportunities during the
search.  Therefore, we add a reduction phase \texttt{ReduceEq} between
stabilization steps, exhaustively applying reduction rules defined below until none of
them is applicable. In particular, every equality produced by \textsc{Subst},
\textsc{Refine}, or an overlap transformation is reduced before it can be
selected as the next stability predicate.

To exploit length information during reduction, we define the
\emph{symbolic length} of a string term $s$ relative to the current regular
assignment $\Reg$ as
\[
  \texttt{SymLen}(s)=
  \begin{cases}
    0, & \text{if }s=\varepsilon,\\
    |w|, & \text{if }s=x\text{ and }\Reg(x)=\{w\},\\
    \mathrm{len}_x, & \text{if }s=x\text{ and }\Reg(x)\text{ is not a singleton},\\
    \texttt{SymLen}(x)+\texttt{SymLen}(s'), & \text{if }s=xs'.
  \end{cases}
\]
For string terms $s$ and $t$, define their \emph{symbolic-length difference}
by $\Delta(s,t)\coloneqq\texttt{SymLen}(s)-\texttt{SymLen}(t)$.
We write
$\Delta(s,t)\mathord{\downarrow}\coloneqq c+\sum_{x\in\Var(st)}a_x\mathrm{len}_x$ for the result
of expanding $\Delta(s,t)$ and collecting like terms, where $c\in\mathbb{Z}$
and $a_x\in\mathbb{Z}$ are the constant term and the coefficient
of $\mathrm{len}_x$, respectively.

\textsc{LenConflict} closes a branch if the constant term has the same
sign as every nonzero coefficient of the symbolic-length difference.
Since string lengths are nonnegative, the difference cannot be zero,
contradicting the equality.
\[
\begin{gathered}
 \textsc{LenConflict}\;
 \dfrac{(\mathsf{N},s= t::\mathsf{F},\Reg,\SEq,\LVar)}
 {\varnothing}\;
 \varphi_{\textsc{LenConflict}}
 \\[0.6ex]
 \begin{aligned}
 \varphi_{\textsc{LenConflict}}
 \overset{\mathrm{def}}{\Longleftrightarrow}{}&
   \Delta(s,t)\mathord{\downarrow}=c+\sum_{x\in\Var(st)}a_x\mathrm{len}_x,\\[-0.4ex]
 &c\ne0,\quad
   \text{$c$ and all nonzero coefficients $a_x$ have the same sign}.
 \end{aligned}
\end{gathered}
\]

When the constant term is zero and all nonzero coefficients have the same
sign, equality of the symbolic lengths forces every variable with a
nonzero coefficient to have zero length.  \textsc{LenZero} captures this
consequence by equating their concatenation with $\varepsilon$.
\[
\begin{gathered}
 \textsc{LenZero}\;
 \dfrac{(\mathsf{N},s= t::\mathsf{F},\Reg,\SEq,\LVar)}
 {\left(\mathsf{N}\cup\{x_1\cdots x_k=\varepsilon\},
   \mathsf{F}\cup\{s= t,x_1\cdots x_k=\varepsilon\},
   \Reg,\SEq,\LVar\right)}\;
 \varphi_{\textsc{LenZero}}
 \\[0.6ex]
 \begin{aligned}
 \varphi_{\textsc{LenZero}}
 \overset{\mathrm{def}}{\Longleftrightarrow}{}&
   \Delta(s,t)\mathord{\downarrow}=\sum_{x\in\Var(st)}a_x\mathrm{len}_x,\\[-0.4ex]
 &\text{all nonzero coefficients $a_x$ have the same sign},\\[-0.4ex]
 &\{x\in\Var(st)\mid a_x\ne0\}
   =\{x_i\mid 1\leq i\leq k\},\quad k>0,\\[-0.4ex]
 &(x_1\cdots x_k=\varepsilon)\notin\mathsf{N}.
 \end{aligned}
\end{gathered}
\]

\textsc{PropEps} propagates an equality $s=\varepsilon$ by introducing
$x=\varepsilon$ for each variable $x$ occuring in $s$.  These equalities
can subsequently be handled by \textsc{Subst}, which eliminates the
corresponding variables and records their bindings in $\SEq$ when applicable.
\[
 \textsc{PropEps}\;
 \dfrac{(\mathsf{N}\uplus\{uxv=\varepsilon\},\mathsf{F},
   \Reg,\SEq,\LVar)}
 {\left(\mathsf{N}\cup\{uxv=\varepsilon,x=\varepsilon\},
   \mathsf{F}\cup\{x=\varepsilon\},
   \Reg,\SEq,\LVar\right)}\;
 (x=\varepsilon)\notin\mathsf{N}
\]

\textsc{RemoveTrivial} removes a trivial equality $s=s$
from both the constraint set and the processing frontier.
\[
 \textsc{RemoveTrivial}\;
 \dfrac{(\mathsf{N}\uplus\{s=s\},s=s::\mathsf{F},\Reg,\SEq,\LVar)}
 {\left(\mathsf{N},\mathsf{F},
   \Reg,\SEq,\LVar\right)}\;
\]

Finally, \textsc{SplitEq} decomposes $st=uv$ into $s=u$ and
$t=v$ when either the prefixes $s,u$ or the suffixes $t,v$ have equal
symbolic lengths.  We allow a split point to fall inside a fixed word $w$
by factoring it as $w=w_1w_2$ so that the resulting parts have matching
symbolic lengths.
\[
\begin{gathered}
 \textsc{SplitEq}\;
 \dfrac{(\mathsf{N}\uplus\{st=uv\},st=uv::\mathsf{F},\Reg,\SEq,\LVar)}
 {\left(\mathsf{N}\cup\{s=u,t=v\},\mathsf{F}\cup\{s=u,t=v\},
   \Reg,\SEq,\LVar\right)}\;
 \varphi_{\textsc{SplitEq}}
 \\[0.6ex]
 \begin{aligned}
 \varphi_{\textsc{SplitEq}}\overset{\mathrm{def}}{\Longleftrightarrow}{}&
   \Delta(s,u)\mathord{\downarrow}=0\text{ or }\Delta(t,v)\mathord{\downarrow}=0,\ |su|>0,\ |tv|>0.
 \end{aligned}
\end{gathered}
\]

\begin{restatable}
  {lemma}{EqualityAware}
\label{lem:equality-aware-soundness}
Let $S$ be a state and $\mathcal R$ an inference rule
applicable to $S$. Suppose that $\mathcal R$ is
an equality-reduction rule or a stabilization
processing an
equality predicate. Then $\bigvee_{S'\in\mathcal R(S)}\Phi(S')$
is an existential extension of $\Phi(S)$.
\end{restatable}

\begin{proof}
    See Appendix \ref{sec:proof-equality-aware}.
\end{proof}


\subsection{Loop Reasoning}
\label{sec:cyclic-reasoning}
The equality-aware stabilization procedure described above can expose
cyclic dependencies that are initially distributed across several
equations as local overlaps between variable bindings. Such overlaps
are problematic for ordinary stabilization: repeatedly decomposing the
equations may introduce new residual variables while preserving the
same dependency pattern, leading to an unbounded sequence of refinement
steps. We introduce dedicated loop reasoning rules to handle such
bindings by applying the classical decomposition from~\cite{lyndon1962equation}
and replacing the recurring decomposition with
a finite system of word equations and regular constraints that can
be processed by the existing stabilization procedure.

\begin{definition}[Overlapping bindings]
\label{def:overlapping-bindings}
Two word equations $e_1,e_2$ are said to be \emph{overlapping bindings} for a variable
$x$ with \emph{overlap variable} $x'$ if $\{e_1,e_2\}$ is in one of
the following forms:
\[
  \{x\approx sx',\;x't\approx x\}
  \qquad\text{or}\qquad
  \{x\approx x't,\;sx'\approx x\},
\]
where $|s|,|t|\ge 1$, $x\neq x'$, and
$\{x,x'\}\cap(\Var(s)\cup\Var(t))=\varnothing$.
\end{definition}

\begin{example}
Consider the equality $sxu=xt$ obtained in
Example~\ref{ex:expose}. On an alignment branch where
$|\nu(x)|\geq|\nu(s)|$, concatenation-eliminating
refinement may introduce fresh variables $x'$ and $t'$
and yield the binding equations $x=sx', x=x't', t=t'u$.
Here, $x'$ represents the common fragment of the two
occurrences of $x$. Equations underlying the first two predicates form
overlapping bindings for $x$ with overlap variable $x'$.
Eliminating $x$ from these bindings yields
$sx'=x't'$, revealing a periodic dependency that may
persist under further refinement. Rather than repeatedly
refining this equation, \textsc{OvPre} can process the
overlapping bindings directly when its side condition holds.
\end{example}

Overlapping bindings need not occur in the original input: as illustrated
in the example above,
the
refinement steps can expose them from dependencies distributed
across multiple equations. Rather than allowing this pattern to recur
indefinitely, we introduce specialized loop reasoning rules to
recognize this local
structure and record its periodic consequences before any further refinement.
When $\nu(s)\ne\varepsilon$, the induced equation
$sx'=x't$ exhibits a periodic structure characterized by the following
classical factorization.
\begin{lemma}[{\cite[Lemma~2]{lyndon1962equation}}]
\label{lem:classical-overlap}
If $a,b,c\in\Sigma^*$ satisfy
$ab=bc$ and $a\ne \varepsilon$, then
$a=uv,\ b=(uv)^ku,$ and $c=vu$ for some $u,v\in \Sigma^*$ and $k\ge 0$.
\end{lemma}
By Lemma~\ref{lem:classical-overlap}, the solution set of $x\approx sx'\land x't\approx x\land \Reg$ is
\[
\fitcondition{
  \begin{aligned}
    &\{\nu\in[\![\Reg]\!]\mid \nu(s)=\nu(t)=\varepsilon, \nu(x)=\nu(x')\}\ \cup \\
    & \quad\{\nu\in[\![\Reg]\!]\mid\exists k\ge 0, u,v\in\Sigma^*, \nu(x)=\nu(sx'), \nu(x')=(uv)^ku,\nu(s)=uv,\nu(t)=vu\}.
  \end{aligned}
}
\]
Introducing fresh variables $x_1$ for $(uv)^k$, $x_2$ for $u$ and $r$ for $v$
makes the factorization explicit.  The solution additionally
satisfies $\nu(x_1)=\nu(s^k)$.  Unlike existing methods~\cite{eisenhofer2025solving,liang2016efficient},
we deliberately avoid encoding this power relation.
Instead, we retain its consequences $sx_1\approx x_1s$ and $x_1\in\Reg(s)^*$.
Although two constraints alone do not suffice to characterize this relation,
together with the remaining factorization equations they still imply the
original overlapping bindings. 
This yields the rule \textsc{OvPre}.
\[
\begin{gathered}
  \fitcondition{
  \textsc{OvPre}\;
  \dfrac{(\mathsf{N}\uplus\{x=sx',x't= x\},
    \langle x=sx',x't= x\rangle::\mathsf{F},\Reg,\SEq,\LVar)}
  {\{(\mathsf{N}\cup\mathsf{G}_{\varepsilon},
    \mathsf{F}\cup\mathsf{G}_{\varepsilon},\Reg,\SEq,\LVar), (\mathsf{N}\cup\mathsf{G},
    \mathsf{F}\cup\mathsf{G},\Reg\cup\{x_1\mapsto \Reg(s)^*\},\SEq,\LVar')\}}
  \;\varphi_{\mathrm{pre}}
  }\\[-0.2ex]
  \fitcondition{
    \begin{aligned}
    &\varphi_{\mathrm{pre}}\overset{\mathrm{def}}{\Longleftrightarrow}
      {}\{x\approx sx',\ x't\approx x\}
        \text{ are overlapping bindings},\ x'\notin\Var(\mathsf{N}).\\[-0.4ex]
    &\text{where }\mathsf{G}_{\varepsilon}
      =\{st=\varepsilon, x=x'\},
      \mathsf{G}=\{x'= x_1x_2,\;
        x=sx',\;
        s= x_2r,\;
        t= rx_2,\;
        sx_1= x_1s\},\\[-0.4ex]
    &\phantom{\text{where }}
      \LVar'=\LVar\cup\{x_1,x_2\}\text{ if }x'\in\LVar,\quad
      \LVar'=\LVar\text{ otherwise}.
    \end{aligned}
  }
\end{gathered}
\]
The rule \textsc{OvSuf} for overlapping bindings in the other form is obtained symmetrically, and its
correctness follows from a similar argument.
\[
\begin{gathered}
  \fitcondition{
  \textsc{OvSuf}\;
  \dfrac{(\mathsf{N}\uplus\{x=x't,sx'= x\},
    \langle x=x't,sx'= x\rangle::\mathsf{F},\Reg,\SEq,\LVar)}
  {\{(\mathsf{N}\cup\mathsf{G}_{\varepsilon},
    \mathsf{F}\cup\mathsf{G}_{\varepsilon},\Reg,\SEq,\LVar), (\mathsf{N}\cup\mathsf{G},
    \mathsf{F}\cup\mathsf{G},\Reg\cup\{x_1\mapsto \Reg(t)^*\},\SEq,\LVar')\}}
  \;\varphi_{\mathrm{suf}}
  }\\[-0.2ex]
  \fitcondition{
    \begin{aligned}
    &\varphi_{\mathrm{suf}}\overset{\mathrm{def}}{\Longleftrightarrow}
      {}\{x\approx x't,\ sx'\approx x\}
        \text{ are overlapping bindings},\ x'\notin\Var(\mathsf{N})\\[-0.4ex]
    &\text{where }\mathsf{G}_{\varepsilon}
      =\{st=\varepsilon, x=x't\},
      \mathsf{G}=\{x'= x_2x_1,\;
        x=x't,\;
        s= x_2r,\;
        t= rx_2,\;
        tx_1= x_1t\},\\[-0.4ex]
    &\phantom{\text{where }}
      \LVar'=\LVar\cup\{x_1,x_2\}\text{ if }x'\in\LVar,\quad
      \LVar'=\LVar\text{ otherwise}.
    \end{aligned}
  }
\end{gathered}
\]

\paragraph{Specialization for fixed words.}
When $s$ or $t$ has a singleton regular language under $\Reg$,
the overlapping bindings admits a more useful factorization for stabilization.
For example, suppose $\Reg(s)=\{w\}$ for some nonempty word $w$.
By Lemma~\ref{lem:classical-overlap}, the solution set of 
$x\approx sx'\land x't\approx x\land \Reg$ is
\[
  \{\nu\in[\![\Reg]\!]\mid\exists k\ge 0, u\in\Sigma^*, \nu(x)=w^{k+1}u, \nu(x')=w^ku,u\nu(t)=wu,|u|<|w|\}
\]
Introducing fresh variables $x_1$ for $w^k$ and $x_2$ for $u$ gives the factorization
\[
  \exists x_1\ \exists x_2\;
  \bigl(
     x\approx wx'
      \land x'\approx x_1x_2
      \land x_2t\approx wx_2
     \land x_1\in\{w\}^*
      \land x_2\in\operatorname{Pre}_{<}(w)
      \land\Reg
  \bigr),
\]
where $  \operatorname{Pre}_{<}(w)
  =\{u\in\Sigma^*\mid
      \exists v\in\Sigma^+,\;w=uv\}$.
Unlike the general factorization, this representation
directly bounds the length of the residual $x_2$ by the regular constraint.
Consequently, the remaining predicate
$x_2t= wx_2$ cannot reproduce the overlap through $x_2$.
We use this factorization as a specialized instantiation
of \textsc{OvPre}, and denote it by \textsc{OvWord}.
\[
\begin{gathered}
  \fitcondition{
  \textsc{OvWord}\;
  \dfrac{(\mathsf{N}\uplus\{x=sx',x't= x\},
    \langle x=sx',x't= x\rangle::\mathsf{F},\Reg,\SEq,\LVar)}
  {(\mathsf{N}\cup\mathsf{G}_w,
    \mathsf{F}\cup\mathsf{G}_w,
    \Reg\cup\{x_1\mapsto\{w\}^*,
      x_2\mapsto\operatorname{Pre}_{<}(w)\},
    \SEq,\LVar')}
  \;\varphi_{\textsc{OvWord}}
  }\\[-0.2ex]
  \fitcondition{
    \begin{aligned}
    &\varphi_{\textsc{OvWord}}\overset{\mathrm{def}}{\Longleftrightarrow}
      {}\{x\approx sx',\ x't\approx x\}
        \text{ are overlapping bindings},\,\\[-0.4ex]
    & \;\quad\qquad\qquad x'\notin\Var(\mathsf{N}),\,
        \Reg(s)=\{w\},\,w\in\Sigma^+.\\[-0.4ex]
    &\text{where }\mathsf{G}_w=\{x'= x_1x_2, x=sx',
        x_2t= wx_2\},\\[-0.4ex]
    &\phantom{\text{where }}
      \LVar'=\LVar\cup\{x_1,x_2\}\text{ if }x'\in\LVar,
      \quad \LVar'=\LVar\text{ otherwise}.
    \end{aligned}
  }
\end{gathered}
\]
The case $\Reg(t)=\{w\}$ admits an analogous specialization
by factoring the residual with respect to $t$.
The two corresponding specializations for the other overlap
form follow symmetrically.
We omit these three analogous rules.
\begin{restatable}
  {lemma}{LoopReasoning}
\label{lem:loop-reasoning-soundness}
Let $\mathcal{R}$ be any of the loop-reasoning rules.
Then $\bigvee_{S'\in\mathcal{R}(S)}\Phi(S')$
is an existential extension of $\Phi(S)$ for every state $S$ to which $\mathcal{R}$
applies.
\end{restatable}
\begin{proof}
See Appendix \ref{sec:proof-loop-reasoning}.
\end{proof}

\subsection{Correctness Guarantees}
Finally, we show that the extended decision procedure preserves the
completeness and termination properties on chain-free inputs, and
is sound on arbitrary inputs.

\begin{theorem}[Termination and correctness for chain-free inputs]
\label{thm:chain-free-guarantees}
For every chain-free input $\varphi$, the presented decision procedure
terminates and returns $\sat$ if and only if $\varphi$ is satisfiable.
\end{theorem}
\begin{proof}[Sketch]
By the construction of $\mathtt{incl}$, chain-free inputs
yield no pairs of opposite inclusions, so
$\mathtt{merge}$ introduces no equality predicates.
Since the adapted stabilization rules preserve the
original processing of inclusion predicates and cannot
introduce equalities from an inclusion-only state,
the equality-specific and loop-reasoning rules never
apply. The procedure therefore coincides with the
original stabilization procedure on this fragment,
and the claim follows directly from its termination
and correctness guarantees~\cite{chocholaty2026string}.
\end{proof}

\begin{restatable}[Soundness for general inputs]
    {theorem}{GeneralSoundness}
\label{thm:general-soundness}
For any input $\varphi$, if the presented procedure returns
$\sat$ (respectively, $\unsat$), then $\varphi$ is satisfiable
(respectively, unsatisfiable).
\end{restatable}
\begin{proof}
See Appendix \ref{sec:proof-general-soundness}.
\end{proof}

\section{Evaluation}
\label{sec:evaluation}
\paragraph{Experimental setup.}
We implemented \solvername{} by extending Z3-Noodler
1.6.1 (based on Z3 4.16.0) with the proposed techniques and made the source code publicly
available on Github\footnote{\url{https://github.com/Wu-02/z3-noodler}}. We compared
it against the unmodified Z3-Noodler, and other state-of-the-art string
solvers cvc5 1.3.4 and OSTRICH 2.1.
To assess the contribution of the proposed loop reasoning,
we also evaluated a variant of \solvername{} with loop reasoning
disabled, labeled ``no loop reasoning'' in the tables.
An instance is considered solved only if the solver returns
$\sat$ or $\unsat$. For each solver, we report the number of solved
instances and the total CPU time spent on these instances.
We also cross-checked the satisfiability results reported by all evaluated
solvers and found no conflicts.
All experiments were conducted in
an Ubuntu 25.04 Docker container on a machine equipped with an Intel Xeon
Gold 6240 CPU at 2.60 GHz. Each solver run was limited to one CPU
core, 8 GiB of memory, and a timeout of 120 seconds.

\subsection{Diagnostic evaluation}
To expose the limitations of existing solvers in handling
cyclic variable dependencies and demonstrate the benefits
of our extensions, we construct a diagnostic benchmark
of systematically generated, syntactically simple word
equations combined with regular constraints.
We first enumerate equations with exactly three components on
each side, drawn from the variables $x,y$ and the letters
\texttt{a}, \texttt{b}.
We require each side to contain both a variable and a letter,
with one variable occurring exactly three times in total
and the other exactly once.
These restrictions yield non-quadratic equations in which the
variable occurring three times necessarily appears on both sides,
inducing a cyclic dependency.
We discard equations whose sides begin or end with the same
component or distinct letters to exclude immediate
cancellation and boundary mismatches.
To eliminate syntactic duplicates, we retain one
representative per equivalence class under variable
renaming, swapping the two letters, and exchanging
the sides, yielding 16 equation templates.
For each template, we generate four instances by adding
the regular constraints $\operatorname{contains}(x,p)$ and
$\operatorname{contains}(y,q)$, where
$p,q\in\{\texttt{a},\texttt{b}\}$ and
$\operatorname{contains}(z,c)$ requires $z$ to contain
the character $c$.
For example, one such instance is
$x\texttt{a}x\approx\texttt{a}xy
\land\operatorname{contains}(x,\texttt{a})
\land\operatorname{contains}(y,\texttt{b})$.
The resulting 64 instances comprise 12 satisfiable
and 52 unsatisfiable cases.

\begin{table}[t]
  \centering
  \caption{Numbers of solved and unresolved instances and
    total CPU time (in seconds) for each solver on the
    64 diagnostic instances (12 SAT and 52 UNSAT).
    CPU time is summed over solved instances only.}
  \label{tab:crafted-results}
  \setlength{\tabcolsep}{6pt}
  \renewcommand{\arraystretch}{1.12}
  \begin{tabular}{cccccc}
    \toprule
    Solver & SAT (12) & UNSAT (52) & solved & unresolved & time (s) \\
    \midrule
    \solvername{} & \textbf{12} & \textbf{52} & \textbf{64} & \textbf{0} & 1.784 \\
    no loop reasoning & 12 & 20 & 32 & 32 & 0.763 \\
    Z3-Noodler & 12 & 0 & 12 & 52 & 0.340 \\
    cvc5 & 12 & 8 & 20 & 44 & 0.354 \\
    OSTRICH2 & 12 & 0 & 12 & 52 & 0.385 \\
    \bottomrule
  \end{tabular}
\end{table}

Table~\ref{tab:crafted-results} compares the solvers
on the 64 diagnostic instances.
All solvers solve the 12 satisfiable instances, but
their performance differs substantially on the 52
unsatisfiable instances.
Neither Z3-Noodler nor OSTRICH2 solves any of the
unsatisfiable instances, while cvc5 solves only 8, and
all unresolved runs hit either the time or memory limit.
In contrast, \solvername{} is able to solve all 64 instances in less
than 2 seconds.
These results show that even syntactically small
string constraints involving cyclic variable dependencies can
pose significant challenges to
existing solvers, while the proposed extensions
effectively address such cases.
Furthermore, disabling loop reasoning reduces the number of
instances solved by \solvername{} from 64 to 32 with only 20 unsatisfiable
instances solved, demonstrating the substantial contribution
of the loop reasoning rules to this improvement.

\subsection{SMT-LIB benchmark evaluation}
To complement the diagnostic evaluation, we assess the impact of the
proposed extensions on solver coverage and runtime on established
benchmarks, where Z3-Noodler already achieves high coverage.
We evaluate \solvername{} on eight
benchmark sets selected from the SMT-LIB'25 release~\cite{smtlib2025benchmarks}.
Following the benchmark classification used in previous evaluations
of Z3-Noodler~\cite{chocholaty2025z3},
we focus on the equation-oriented benchmarks, which primarily
contain word equations and length constraints, with some regular
constraints and other string operations. These benchmarks are
particularly relevant to our extensions, whereas
benchmark sets dominated by regular constraints or extended
string operations primarily exercise components inherited
unchanged from Z3-Noodler and are thus excluded from this evaluation.

\begin{table}[t]
  \centering
  \caption{Numbers of solved instances and total CPU time
    (in seconds) for each solver on eight equation-oriented
    benchmark sets from SMT-LIB'25. Numbers in parentheses indicate the
    total number of instances in each set; CPU time is
    summed over solved instances only.}
  \label{tab:main-results}
  \setlength{\tabcolsep}{5pt}
  \renewcommand{\arraystretch}{1.12}
  \sbox{\mainresultstablebox}{%
    \begin{tabular}{c *{5}{cc}}
      \toprule
      & \multicolumn{2}{c}{\resultheader{\textbf{Slog}}{1,976}}
      & \multicolumn{2}{c}{\resultheader{\textbf{Slent}}{1,128}}
      & \multicolumn{2}{c}{\resultheader{\textbf{Omark}}{17}}
      & \multicolumn{2}{c}{\resultheader{\textbf{Woorpje}}{809}}
      & \multicolumn{2}{c}{\resultheader{\textbf{Kepler}}{587}} \\
      \cmidrule(lr){2-3}\cmidrule(lr){4-5}\cmidrule(lr){6-7}\cmidrule(lr){8-9}\cmidrule(lr){10-11}
      & \multicolumn{1}{c}{solved} & \multicolumn{1}{c}{time (s)}
      & \multicolumn{1}{c}{solved} & \multicolumn{1}{c}{time (s)}
      & \multicolumn{1}{c}{solved} & \multicolumn{1}{c}{time (s)}
      & \multicolumn{1}{c}{solved} & \multicolumn{1}{c}{time (s)}
      & \multicolumn{1}{c}{solved} & \multicolumn{1}{c}{time (s)} \\
      \midrule
      \solvername{} & \textbf{1,976} & 49.341
        & \textbf{1,128} & 44.283 & \textbf{13} & 0.349 & \textbf{795} & 226.097 & \textbf{587} & 28.275 \\
      no loop reasoning & \textbf{1,976} & 50.103
        & \textbf{1,128} & 44.631 & 11 & 0.265 & \textbf{795} & 222.809 & 586 & 24.427 \\
      Z3-Noodler & \textbf{1,976} & 46.930
        & \textbf{1,128} & 55.331 & 10 & 0.260 & 781 & 281.013 & 584 & 23.550 \\
      cvc5 & 1,975 & 34.207 & 1,104 & 1,074.377 & 5 & 1.131 & 755 & 85.974 & 349 & 6.703 \\
      OSTRICH2 & 1,660 & 1,910.877 & 977 & 369.521 & 1 & 0.032 & 694 & 2,001.505 & 187 & 9.689 \\
      \bottomrule
    \end{tabular}%
  }
  \resizebox{\textwidth}{!}{\usebox{\mainresultstablebox}}

  \medskip

  \resizebox{\textwidth}{!}{%
    \begin{tabular*}{\wd\mainresultstablebox}{c @{\extracolsep{\fill}} *{4}{cc}}
      \toprule
      & \multicolumn{2}{c}{\resultheader{\textbf{Webapp}}{681}}
      & \multicolumn{2}{c}{\resultheader{\textbf{Kaluza}}{19,432}}
      & \multicolumn{2}{c}{\resultheader{\textbf{Norn}}{1,027}}
      & \multicolumn{2}{c}{\resultheader{$\boldsymbol{\operatorname{\Sigma}}$}{25,657}} \\
      \cmidrule(lr){2-3}\cmidrule(lr){4-5}\cmidrule(lr){6-7}\cmidrule(lr){8-9}
      & \multicolumn{1}{c}{solved} & \multicolumn{1}{c}{time (s)}
      & \multicolumn{1}{c}{solved} & \multicolumn{1}{c}{time (s)}
      & \multicolumn{1}{c}{solved} & \multicolumn{1}{c}{time (s)}
      & \multicolumn{1}{c}{solved} & \multicolumn{1}{c}{time (s)} \\
      \midrule
      \solvername{} & \textbf{681} & 56.108
        & 19,406 & 792.355 & \textbf{1,027} & 24.814
        & \textbf{25,613} & 1,221.622 \\
      no loop reasoning & \textbf{681} & 55.944
        & 19,406 & 811.699 & \textbf{1,027} & 25.125
        & 25,610 & 1,235.003 \\
      Z3-Noodler & \textbf{681} & 169.977
        & 19,406 & 736.335 & \textbf{1,027} & 24.361
        & 25,593 & 1,337.757 \\
      cvc5 & 589 & 996.609 & \textbf{19,426} & 910.971 & 943 & 459.988
        & 25,146 & 3,569.960 \\
      OSTRICH2 & 484 & 942.676 & 19,041 & 3,931.585 & \textbf{1,027} & 160.372
        & 24,071 & 9,326.257 \\
      \bottomrule
    \end{tabular*}%
  }
\end{table}

Table~\ref{tab:main-results} summarizes the number of solved instances and the
total CPU time over solved instances for each solver grouped
by benchmark sets. Overall, \solvername{} achieves the highest coverage, leaving
only 44 unresolved instances, which is 20 fewer than Z3-Noodler.
Despite solving more instances, \solvername{} takes only
1,221.6 s of total CPU time, compared with 1,337.8 s for Z3-Noodler.
The improvements
are concentrated
in Omark, Woorpje and Kepler, with identical solved counts on the
remaining sets.
\solvername{} also achieves the highest solved count, including ties, on seven
of the eight benchmark sets.
The exception is Kaluza, where cvc5 solves 20 more instances on which
both Z3-Noodler and our solver return \textsf{unknown}.
Disabling loop
reasoning reduces the number of solved instances by only three, suggesting
that the improvements on these benchmarks primarily stem from equality-aware
stabilization rather than loop reasoning.  Together with the diagnostic
evaluation, these
results show that the proposed techniques can address challenging
cyclic dependencies while preserving the competitive performance
on established benchmarks.

\begin{wrapfigure}{r}{0.4\textwidth}
    \centering
    \includegraphics[width=\linewidth]{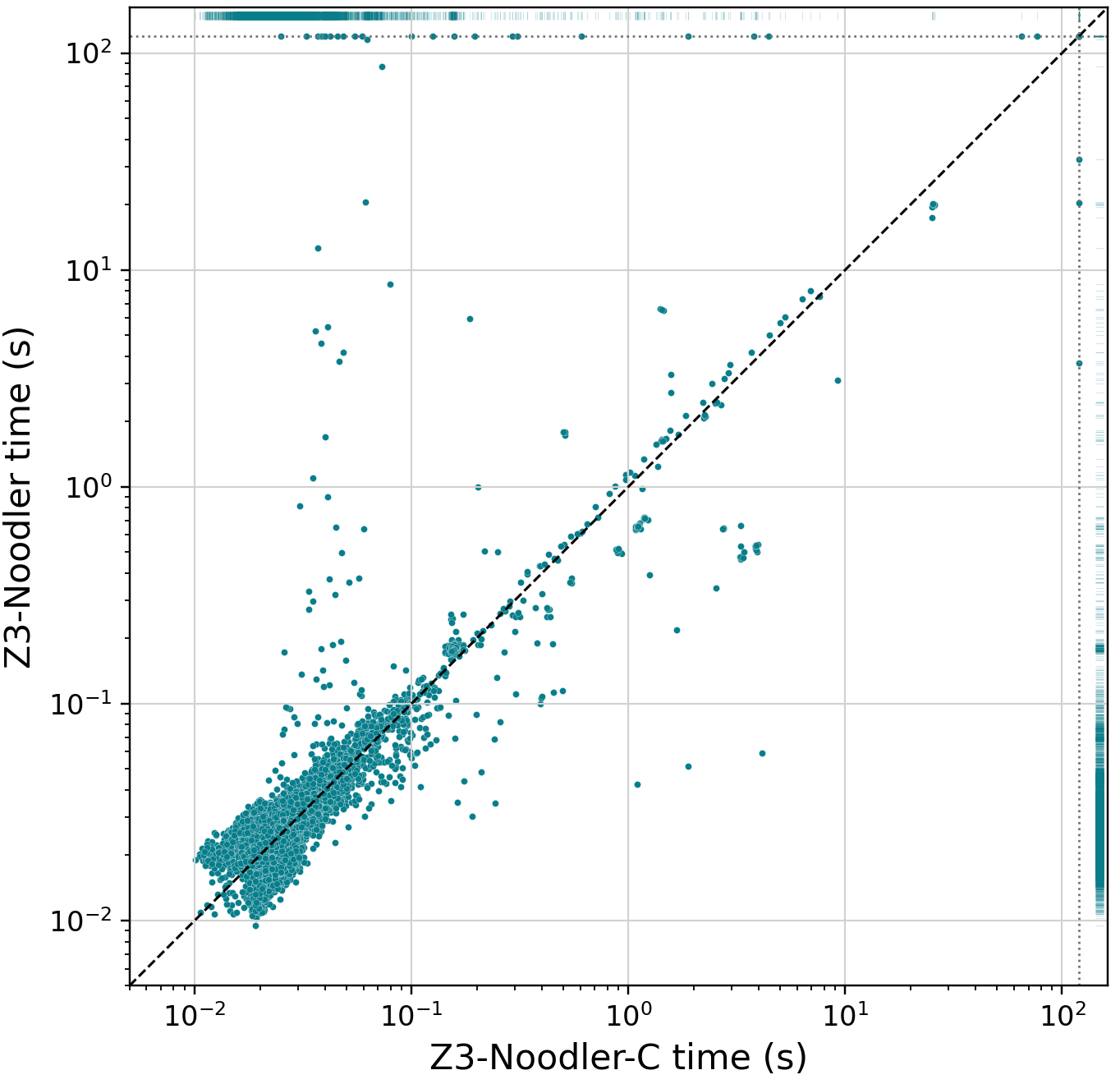}
    \caption{Comparison between \solvername{}
    (x-axis) and Z3-Noodler (y-axis) on the eight SMT-LIB
    benchmark sets. Both axes are logarithmic, with
    unsolved instances assigned 120\,s. Each point represents
    one instance.}
    \label{fig:smtlib-scatter}
\end{wrapfigure}
Figure~\ref{fig:smtlib-scatter} compares \solvername{} 
with the base solver on these benchmark sets. 
A small number of points lie well below the diagonal, since
the more aggressive concatenation-eliminating noodlification of
\solvername{} may introduce additional refinement branches by exposing 
the occurrence boundaries of shared variables that are not length-sensitive,
leading to performance regressions on some
instances that Z3-Noodler can originally solve. However, the overall
impact of this regression is limited, as most points lie close to
the diagonal, indicating that \solvername{} largely
preserves the efficiency of Z3-Noodler on the majority of instances.
Meanwhile, a noticeable number of points lie well above the diagonal, 
indicating substantial speedups on certain instances. More importantly, the 
concentration  of points along the $y=120\,\mathrm{ s}$ line above the 
diagonal further highlights the improved solving capability of \solvername{}, 
especially on instances that Z3-Noodler fails to solve within the time limit.
Overall, these results suggest that our extensions provide significant 
improvements to Z3-Noodler without introducing significant overhead on most benchmarks.

\section{Related Work}

\paragraph{Stabilization-based string solving.}
Stabilization was first introduced in~\cite{blahoudek2023word} for solving
word equations together with  regular constraints by refining regular
assignments until all relevant equation inclusions in the inclusion graph
become stable.  Subsequent work~\cite{chen2023solving} extends it with 
length reasoning by separating length-aware variables during refinement.
Stabilization-based decision procedure forms the basis of
Z3-Noodler~\cite{chen2024z3} and has subsequently been extended with support
for complementary decision procedures and model
generation~\cite{chocholaty2025z3}, string--integer
conversions~\cite{havlena2024cooking}, position
constraints~\cite{chen2025uniform}, and functional finite-state
transducers~\cite{chocholaty2026string}.  On chain-free~\cite{abdulla2019chain}
inputs, the acyclicity of inclusion graph guarantees termination and
completeness of this procedure.

Unlike the original procedure, which processes opposite inclusions with
the same underlying equation separately, our extension processes them
as a single equality predicate, enabling equation-level reduction
and substitution by considering both inclusion directions together. 
Moreover, we reuse concatenation-eliminating
refinement to expose shared variable occurrences in alignment bindings
even when no length reasoning is required.

\paragraph{Cyclic variable dependencies.}
Various approaches to handle cyclic variable dependencies in string constraint
solving have been proposed. cvc4~\cite{liang2016efficient} introduced the
\textsc{F-Loop} rule to breaks certain potentially infinite derivation loops
caused by repeated word-equation splitting, using fresh variables and symbolic Kleene-star constraints. ZIPT~\cite{eisenhofer2025solving} extends
Nielsen transformations with explicit power terms to summarize repeated substitutions. Both 
approaches exploit the classical factorization of overlapping words established by Lyndon 
and Sch\"utzenberger~\cite{lyndon1962equation}. Kepler$_{22}$~\cite{le2018decision} links reduction tree nodes that represent equivalent word-equation systems under variable substitutions
to summarize potentially unbounded reduction sequences.
OSTRICH2~\cite{hague2025ostrich2} breaks cycles by deriving emptiness constraints from length reasoning and eliminating redundant 
equations. Z3-Noodler~\cite{chen2024z3} employs inclusion-based stabilization, which is
particularly effective for chain-free constraints~\cite{abdulla2019chain} but
does not fully exploit the underlying word equations when cyclic dependencies are
present. Our solver extends Z3-Noodler's automata-based stabilization with various word-based reasoning techniques, achieving better overall results in our evaluation.

\section{Conclusion}
We have implemented \solvername{} on top of Z3-Noodler by extending the
stabilization-based decision procedure with equality predicates and specialized
loop reasoning rules. Our evaluation shows that \solvername{} substantially
improves coverage on the diagnostic benchmarks while remaining competitive
on established SMT-LIB benchmarks.
These results suggest that the extensions can considerably strengthen stabilization-based
solving in the presence of cyclic variable dependencies 
while retaining its effectiveness on broader benchmarks.

\bibliographystyle{splncs04}
\bibliography{refs}

@inproceedings{abdulla2019chain,
  title={Chain-free string constraints},
  author={Abdulla, Parosh Aziz and Atig, Mohamed Faouzi and Diep, Bui Phi and Hol{\'\i}k, Luk{\'a}{\v{s}} and Jank{\r{u}}, Petr},
  booktitle={International Symposium on Automated Technology for Verification and Analysis},
  pages={277--293},
  year={2019},
  organization={Springer}
}

@inproceedings{blahoudek2023word,
  title={Word equations in synergy with regular constraints},
  author={Blahoudek, Franti{\v{s}}ek and Chen, Yu-Fang and Chocholat{\`y}, David and Havlena, Vojt{\v{e}}ch and Hol{\'\i}k, Luk{\'a}{\v{s}} and Leng{\'a}l, Ond{\v{r}}ej and S{\'\i}{\v{c}}, Juraj},
  booktitle={International Symposium on Formal Methods},
  pages={403--423},
  year={2023},
  organization={Springer}
}

@inproceedings{chen2024z3,
  title={Z3-noodler: An automata-based string solver},
  author={Chen, Yu-Fang and Chocholat{\`y}, David and Havlena, Vojt{\v{e}}ch and Hol{\'\i}k, Luk{\'a}{\v{s}} and Leng{\'a}l, Ond{\v{r}}ej and S{\'\i}{\v{c}}, Juraj},
  booktitle={International Conference on Tools and Algorithms for the Construction and Analysis of Systems},
  pages={24--33},
  year={2024},
  organization={Springer}
}

@article{chen2023solving,
  title={Solving string constraints with lengths by stabilization},
  author={Chen, Yu-Fang and Chocholat{\`y}, David and Havlena, Vojt{\v{e}}ch and Hol{\'\i}k, Luk{\'a}{\v{s}} and Leng{\'a}l, Ond{\v{r}}ej and S{\'\i}{\v{c}}, Juraj},
  journal={Proceedings of the ACM on Programming Languages},
  volume={7},
  number={OOPSLA2},
  pages={2112--2141},
  year={2023},
  publisher={ACM New York, NY, USA}
}

@inproceedings{chocholaty2025z3,
  title={Z3-noodler 1.3: shepherding decision procedures for strings with model generation},
  author={Chocholat{\`y}, David and Havlena, Vojt{\v{e}}ch and Hol{\'\i}k, Luk{\'a}{\v{s}} and Hrani{\v{c}}ka, Jan and Leng{\'a}l, Ond{\v{r}}ej and S{\'\i}{\v{c}}, Juraj},
  booktitle={International Conference on Tools and Algorithms for the Construction and Analysis of Systems},
  pages={23--44},
  year={2025},
  organization={Springer}
}

@inproceedings{havlena2024cooking,
  title={Cooking string-integer conversions with noodles},
  author={Havlena, Vojt{\v{e}}ch and Hol{\'\i}k, Luk{\'a}{\v{s}} and Leng{\'a}l, Ond{\v{r}}ej and S{\'\i}{\v{c}}, Juraj},
  booktitle={27th International Conference on Theory and Applications of Satisfiability Testing (SAT 2024)},
  pages={14--1},
  year={2024},
  organization={Schloss Dagstuhl--Leibniz-Zentrum f{\"u}r Informatik}
}

@article{liang2016efficient,
  title={An efficient SMT solver for string constraints},
  author={Liang, Tianyi and Reynolds, Andrew and Tsiskaridze, Nestan and Tinelli, Cesare and Barrett, Clark and Deters, Morgan},
  journal={Formal Methods in System Design},
  volume={48},
  number={3},
  pages={206--234},
  year={2016},
  publisher={Springer}
}

@article{lyndon1962equation,
  title={The equation {$a^M= b^N c^P$} in a free group},
  author={Lyndon, Roger C and Sch{\"u}tzenberger, Marcel-Paul and others},
  journal={Michigan Math. J},
  volume={9},
  number={4},
  pages={289--298},
  year={1962}
}

@inproceedings{eisenhofer2025solving,
  title={On Solving String Equations via Powers and Parikh Images},
  author={Eisenhofer, Clemens and Seiser, Theodor and Bj{\o}rner, Nikolaj and Kov{\'a}cs, Laura},
  booktitle={International Conference on Automated Reasoning with Analytic Tableaux and Related Methods},
  pages={120--138},
  year={2025},
  organization={Springer Nature Switzerland Cham}
}

@inproceedings{chocholaty2026string,
  title={String Solving with Stabilization and Transducers},
  author={Chocholat{\`y}, David and Havlena, Vojt{\v{e}}ch and Hol{\'\i}k, Luk{\'a}{\v{s}} and S{\'\i}{\v{c}}, Juraj and {\v{S}}ed{\`y}, Michal},
  booktitle={International Conference on Computer Aided Verification},
  pages={50--74},
  year={2026},
  organization={Springer}
}

@article{chen2025uniform,
  title={A uniform framework for handling position constraints in string solving},
  author={Chen, Yu-Fang and Havlena, Vojt{\v{e}}ch and He{\v{c}}ko, Michal and Hol{\'\i}k, Luk{\'a}{\v{s}} and Leng{\'a}l, Ond{\v{r}}ej},
  journal={Proceedings of the ACM on Programming Languages},
  volume={9},
  number={PLDI},
  pages={550--575},
  year={2025},
  publisher={ACM New York, NY, USA}
}

@inproceedings{saxena2010symbolic,
  title={A symbolic execution framework for javascript},
  author={Saxena, Prateek and Akhawe, Devdatta and Hanna, Steve and Mao, Feng and McCamant, Stephen and Song, Dawn},
  booktitle={2010 IEEE Symposium on Security and Privacy},
  pages={513--528},
  year={2010},
  organization={IEEE}
}

@article{chen2019decision,
  title={Decision procedures for path feasibility of string-manipulating programs with complex operations},
  author={Chen, Taolue and Hague, Matthew and Lin, Anthony W and R{\"u}mmer, Philipp and Wu, Zhilin},
  journal={Proceedings of the ACM on Programming Languages},
  volume={3},
  number={POPL},
  pages={1--30},
  year={2019},
  publisher={ACM New York, NY, USA}
}

@inproceedings{barbosa2022cvc5,
  title={cvc5: A versatile and industrial-strength SMT solver},
  author={Barbosa, Haniel and Barrett, Clark and Brain, Martin and Kremer, Gereon and Lachnitt, Hanna and Mann, Makai and Mohamed, Abdalrhman and Mohamed, Mudathir and Niemetz, Aina and N{\"o}tzli, Andres and others},
  booktitle={International Conference on Tools and Algorithms for the Construction and Analysis of Systems},
  pages={415--442},
  year={2022},
  organization={Springer}
}

@inproceedings{hague2025ostrich2,
  title={OSTRICH2: Solver for complex string constraints},
  author={Hague, Matthew and Hu, Denghang and Jez, Artur and Lin, Anthony W and Markgraf, Oliver and R{\"u}mmer, Philipp and Wu, Zhilin},
  booktitle={Proceedings of the 25th Conference on Formal Methods in Computer-Aided Design--FMCAD},
  pages={145--158},
  year={2025}
}

@inproceedings{le2018decision,
  title={A decision procedure for string logic with quadratic equations, regular expressions and length constraints},
  author={Le, Quang Loc and He, Mengda},
  booktitle={Asian Symposium on Programming Languages and Systems},
  pages={350--372},
  year={2018},
  organization={Springer}
}

@article{makanin1986problem,
  title={The problem of solvability of equations in a free semigroup},
  author={Makanin, Gennady S},
  journal={Journal of Symbolic Logic},
  volume={51},
  number={4},
  year={1986}
}

@article{jez2016recompression,
  title={Recompression: a simple and powerful technique for word equations},
  author={Je{\.z}, Artur},
  journal={Journal of the ACM (JACM)},
  volume={63},
  number={1},
  pages={1--51},
  year={2016},
  publisher={ACM New York, NY, USA}
}

@inproceedings{abdulla2014string,
  title={String constraints for verification},
  author={Abdulla, Parosh Aziz and Atig, Mohamed Faouzi and Chen, Yu-Fang and Hol{\'\i}k, Luk{\'a}{\v{s}} and Rezine, Ahmed and R{\"u}mmer, Philipp and Stenman, Jari},
  booktitle={International Conference on Computer Aided Verification},
  pages={150--166},
  year={2014},
  organization={Springer}
}

@inproceedings{wassermann2007sound,
  title={Sound and precise analysis of web applications for injection vulnerabilities},
  author={Wassermann, Gary and Su, Zhendong},
  booktitle={Proceedings of the 28th ACM SIGPLAN Conference on Programming Language Design and Implementation},
  pages={32--41},
  year={2007}
}

@inproceedings{trinh2014s3,
  title={S3: A symbolic string solver for vulnerability detection in web applications},
  author={Trinh, Minh-Thai and Chu, Duc-Hiep and Jaffar, Joxan},
  booktitle={Proceedings of the 2014 ACM SIGSAC Conference on Computer and Communications Security},
  pages={1232--1243},
  year={2014}
}

@inproceedings{eriksson2023black,
  title={Black ostrich: Web application scanning with string solvers},
  author={Eriksson, Benjamin and Stjerna, Amanda and De Masellis, Riccardo and R{\"u}emmer, Philipp and Sabelfeld, Andrei},
  booktitle={Proceedings of the 2023 ACM SIGSAC Conference on Computer and Communications Security},
  pages={549--563},
  year={2023}
}

@inproceedings{backes2018semantic,
  title={Semantic-based automated reasoning for AWS access policies using SMT},
  author={Backes, John and Bolignano, Pauline and Cook, Byron and Dodge, Catherine and Gacek, Andrew and Luckow, Kasper and Rungta, Neha and Tkachuk, Oksana and Varming, Carsten},
  booktitle={2018 Formal Methods in Computer Aided Design (FMCAD)},
  pages={1--9},
  year={2018},
  organization={IEEE}
}

@inproceedings{rungta2022billion,
  title={A billion SMT queries a day},
  author={Rungta, Neha},
  booktitle={International Conference on Computer Aided Verification},
  pages={3--18},
  year={2022},
  organization={Springer}
}

@misc{smtcomp2024strings,
  author = {{SMT-COMP}},
  title  = {{SMT-COMP} 2024: {QF\_Strings} (Single Query Track) Results},
  year   = {2024},
  url    = {https://smt-comp.github.io/2024/results/qf_strings-single-query/},
  note   = {Results generated on 2024-07-08; accessed on 2026-07-23}
}

@misc{smtcomp2025strings,
  author = {{SMT-COMP}},
  title  = {{SMT-COMP} 2025: {QF\_Strings} (Single Query Track) Results},
  year   = {2025},
  url    = {https://smt-comp.github.io/2025/results/qf_strings-single-query/},
  note   = {Results generated on 2025-08-11; accessed on 2026-07-23}
}

@misc{smtlib2025benchmarks,
  author = {Preiner, Mathias and Schurr, Hans-J{\"o}rg and Barrett, Clark and
            Fontaine, Pascal and Niemetz, Aina and Tinelli, Cesare},
  title = {{SMT-LIB} Release 2025 (Non-Incremental Benchmarks)},
  year = {2025},
  url = {https://zenodo.org/records/16740866},
  note = {Version 2025.08.04; accessed on 2026-07-23}
}

@inproceedings{loring2019sound,
  title={Sound regular expression semantics for dynamic symbolic execution of JavaScript},
  author={Loring, Blake and Mitchell, Duncan and Kinder, Johannes},
  booktitle={Proceedings of the 40th ACM SIGPLAN Conference on Programming Language Design and Implementation},
  pages={425--438},
  year={2019}
}

@article{yu2014automata,
  title={Automata-based symbolic string analysis for vulnerability detection},
  author={Yu, Fang and Alkhalaf, Muath and Bultan, Tevfik and Ibarra, Oscar H},
  journal={Formal Methods in System Design},
  volume={44},
  number={1},
  pages={44--70},
  year={2014},
  publisher={Springer}
}

@inproceedings{li2014symjs,
  title={SymJS: automatic symbolic testing of JavaScript web applications},
  author={Li, Guodong and Andreasen, Esben and Ghosh, Indradeep},
  booktitle={Proceedings of the 22nd ACM SIGSOFT International Symposium on Foundations of Software Engineering},
  pages={449--459},
  year={2014}
}

@inproceedings{de2008z3,
  title={Z3: An efficient SMT solver},
  author={De Moura, Leonardo and Bj{\o}rner, Nikolaj},
  booktitle={International conference on Tools and Algorithms for the Construction and Analysis of Systems},
  pages={337--340},
  year={2008},
  organization={Springer}
}

@inproceedings{bjorner2009path,
  title={Path feasibility analysis for string-manipulating programs},
  author={Bj{\o}rner, Nikolaj and Tillmann, Nikolai and Voronkov, Andrei},
  booktitle={International Conference on Tools and Algorithms for the Construction and Analysis of Systems},
  pages={307--321},
  year={2009},
  organization={Springer}
}

@inproceedings{stanford2021symbolic,
  title={Symbolic boolean derivatives for efficiently solving extended regular expression constraints},
  author={Stanford, Caleb and Veanes, Margus and Bj{\o}rner, Nikolaj},
  booktitle={Proceedings of the 42nd ACM SIGPLAN International Conference on Programming Language Design and Implementation},
  pages={620--635},
  year={2021}
}

@inproceedings{trinh2016progressive,
  title={Progressive reasoning over recursively-defined strings},
  author={Trinh, Minh-Thai and Chu, Duc-Hiep and Jaffar, Joxan},
  booktitle={International Conference on Computer Aided Verification},
  pages={218--240},
  year={2016},
  organization={Springer}
}

@inproceedings{liang2014dpll,
  title={A DPLL (T) theory solver for a theory of strings and regular expressions},
  author={Liang, Tianyi and Reynolds, Andrew and Tinelli, Cesare and Barrett, Clark and Deters, Morgan},
  booktitle={International Conference on Computer Aided Verification},
  pages={646--662},
  year={2014},
  organization={Springer}
}

@techreport{barrett2024smtlib,
  author      = {Clark Barrett and Pascal Fontaine and Cesare Tinelli},
  title       = {The {SMT-LIB} Standard: Version 2.7},
  institution = {University of Iowa},
  year        = {2024},
  type        = {Technical Report},
  note        = {Revision 2025-07-07},
  url         = {https://smt-lib.org/papers/smt-lib-reference-v2.7-r2025-07-07.pdf}
}

\newpage
\appendix
\section{Proof of Lemma \ref{lem:equality-aware-soundness}}
\label{sec:proof-equality-aware}
\EqualityAware*
\begin{proof}
\textsc{EqSkip} only removes a predicate from
$\mathsf{F}$ and leaves $\Phi(S)$ unchanged.
For \textsc{Subst}, let $\sigma=\{x\mapsto s\}$
and let $S'$ be the successor state obtained by
applying the rule to $S$.
Every solution $\nu$ of $\Phi(S)$ satisfies
$\nu(x)=\nu(s)$ and hence also satisfies $\Phi(S')$.
Conversely, every solution $\nu$ of $\Phi(S')$ satisfies
$\Phi(S)$ by extending $\nu(x)=\nu(s)$, since
the side condition $x\supseteq_\Reg s$ ensures that
$\nu(x)$ respects the original regular
constraint $\Reg(x)$.
For \textsc{Refine}, the claim follows
from the correctness of the corresponding noodlification
procedures~\cite{blahoudek2023word,chen2023solving}.
Equality reduction also preserves solutions:
\textsc{LenZero} and \textsc{PropEps} only add logical
consequences of $\Phi(S)$, \textsc{RemoveTrivial}
removes a tautology, and \textsc{SplitEq} replaces
an equality with an equivalent pair under its
side condition. Finally, \textsc{LenConflict}
closes only unsatisfiable branches.
Hence, every rule preserves $\Phi(S)$ across its
successor branches, up to existential quantification
of fresh variables.
\end{proof}

\section{Proof of Lemma \ref{lem:loop-reasoning-soundness}}
\label{sec:proof-loop-reasoning}
\LoopReasoning*

\begin{proof}
For \textsc{OvPre}, let $S_\varepsilon$ and
$S_g$ denote its two successor states.
Consider a solution $\nu$ of $\Phi(S)$ satisfying
the overlapping bindings $x=sx'$ and $x't=x$.
If $\nu(s)=\varepsilon$, cancellation yields
$\nu(t)=\varepsilon$ and $\nu(x)=\nu(x')$,
so $\nu$ satisfies $\Phi(S_\varepsilon)$.
Otherwise, by Lemma~\ref{lem:classical-overlap},
there exist $u,v\in\Sigma^*$ and $k\geq 0$ such that
$
\nu(s)=uv,
\nu(t)=vu,
\nu(x')=(uv)^ku.
$
Extending $\nu$ with
$\nu(x_1)=(uv)^k$, $\nu(x_2)=u$, and $\nu(r)=v$
satisfies $\Phi(S_g)$, including
the regular constraint $x_1\in\Reg(s)^*$.
Conversely, every solution of $\Phi(S_\varepsilon)$
satisfies the original bindings, and for any solution
$\nu$ of $\Phi(S_g)$ we have
$
\nu(x't)=\nu(x_1x_2rx_2)=\nu(x_1sx_2)
=\nu(sx_1x_2)=\nu(sx')=\nu(x).
$
Thus, restricting $\nu$ to the original variables
yields a solution of $\Phi(S)$.
The same claim for \textsc{OvSuf} follows
symmetrically by reversing the words and
interchanging the roles of $s$ and $t$.

For \textsc{OvWord}, let $S_w$ denote its successor
state and consider a solution $\nu$ of $\Phi(S)$.
Since $\Reg(s)=\{w\}$ with
$w\ne\varepsilon$, Lemma~\ref{lem:classical-overlap}
gives $\nu(x')=w^ku$ for some $k\geq 0$ and
$u\in\operatorname{Pre}_{<}(w)$ satisfying
$u\nu(t)=wu$.
Extending $\nu$ with $\nu(x_1)=w^k$ and
$\nu(x_2)=u$ satisfies $\Phi(S_w)$, including
the regular constraints on $x_1$ and $x_2$.
Conversely, for any solution $\nu$ of $\Phi(S_w)$,
we have $
\nu(x't)=\nu(x_1x_2t)=\nu(x_1wx_2)
=\nu(wx_1x_2)=\nu(sx')=\nu(x).
$
Hence, restricting $\nu$ to the original variables
satisfies $\Phi(S)$.
The remaining fixed-word specializations
follow analogously. Therefore, for every loop-reasoning
rule $\mathcal R$, $\bigvee_{S'\in\mathcal R(S)}\Phi(S')$
is an existential extension of $\Phi(S)$.
\end{proof}

\section{Proof of Theorem \ref{thm:general-soundness}}
\label{sec:proof-general-soundness}
\paragraph{State invariants.}
The soundness of the procedure relies on the following invariants
to hold for every reachable state
$S=(\mathsf N_S,\mathsf F_S,\Reg_S,\SEq_S,\LVar_S)$:
\begin{enumerate}
    \setlength{\itemindent}{2em}
    \item[(I1)]
    Expanding each equality predicate in $\mathsf N_S$
    into two opposite inclusion predicates produces the
    vertex set of an inclusion graph for $\eq(\mathsf N_S)$,
    and $\mathsf N_S$ does not contain any pair of
    opposite inclusion predicates.

    \item[(I2)]
    Every predicate in $\mathsf N_S\setminus\mathsf F_S$
    is stable under $\Reg_S$, and every variable in $\Var(N_S\setminus\mathsf F_S)\cap\LVar_S$
    only occurs on the right-hand sides of inclusion predicates in $\mathsf N_S\setminus\mathsf F_S$.

    \item[(I3)]
    Every binding in $\SEq_S$ has a distinct variable on
    its left-hand side that does not occur
    in $\mathsf N_S$ or on the right-hand side of any
    binding in $\SEq_S$.
\end{enumerate}

\begin{lemma}
\label{lem:state-invariants}
Every reachable state $S$ satisfies {\rm (I1)--(I3)}.
\end{lemma}

\begin{proof}
We proceed by induction on the length of the derivation.

For the base case, the initial state satisfies
(I1)--(I3) by construction.

For the inductive step, suppose that (I1)--(I3)
hold for a reachable state $S$. We show that they
also hold for every successor state $S'$.
Since \textsc{LenConflict} produces no successor
state, it need not be considered.

\emph{(I1).}
Rules \textsc{Skip} and \textsc{EqSkip} leave
$\mathsf N$ unchanged, while \textsc{RemoveTrivial}
only removes a trivial equality, so they all preserve (I1).

In the following, we omit arguments for the
absence of opposite inclusions by ruling out cycles
involving an inclusion predicate, since any pair
would form such a cycle.
Consider an application of \textsc{CombNL} to
$p=(uvw\supseteq t)$, which replaces $p$ with
$p_1=(uxw\supseteq t)$ and $p_2=(v\supseteq x)$,
where $x$ is fresh. The only new dependency edge
is $p_2\to p_1$.
If this edge formed a cycle, merging
$p_1$ and $p_2$ back into $p$ would yield a
cycle containing $p$ in the original inclusion
graph, contradicting (I1).

For \textsc{Subst} applied to an inclusion
$p=(s\supseteq x)$, (I1) ensures that variables in $s$ occur
on no other inclusion's left-hand side.
Consequently, every new dependency edge introduced
by the substitution corresponds to a path through $p$ in the
original graph, and substituting
$s$ for $x$ cannot introduce repeated
left-hand-side occurrences.
For $p=(x\supseteq s)$, the argument is analogous,
since (I1) ensures that $x$ occurs on no other
inclusion's left-hand side.
Therefore, neither case introduces a cycle
involving an inclusion or two inclusions sharing a variable on their left-hand sides.
Thus, \textsc{Subst} preserves (I1).

For \textsc{Refine} applied to an inclusion
$p=(s\supseteq t)$, each alignment splits $p$
into binding predicates, preserving $s$ on the
left and $t$ on the right, with fresh variables
connecting the two sides.
Any cycle involving these bindings
would therefore induce a cycle through $p$
in the original inclusion graph, contradicting
(I1).

Finally, for rules processing equality predicates, (I1) ensures
that variables occurring in equalities do not appear
on the left-hand sides of inclusions.
These rules introduce only equalities involving existing
equality variables and fresh variables, and any
substitution affects only the right-hand sides of
inclusions.
Thus, no new dependency path from an equality to a
inclusion is introduced, and the
left-hand-side occurrences of
inclusions remain unchanged.
Consequently, no cycle involving a
inclusion is introduced, preserving (I1).

\emph{(I2).}
Rules \textsc{Skip} and \textsc{EqSkip} preserve (I2)
by their side conditions.
Rule \textsc{CombNL} introduces an inclusion
$v\supseteq x$ outside the frontier, which is stable
by $\Reg(x)=\Reg(v)$ and contains no length-aware
variables on its left-hand side.

For \textsc{Subst}, inclusions containing the
substituted variable on their left-hand sides
and equalities containing it on either side
are re-added to \textsf{F}. For the remaining
inclusions, substitution can only restrict their
right-hand sides, preserving stability and
length separation.

For \textsc{Refine}, stabilizing noodlification
only restricts the regular languages of variables
on the right-hand side of an inclusion.
By the topological ordering of inclusions in \textsf{F}, these variables
cannot occur on the left-hand sides of previously
processed inclusions, and no equality has yet
been processed. Hence, stability is preserved.
Concatenation-eliminating noodlification, on the
other hand, only introduces regular assignments
for fresh variables, leaving existing assignments
unchanged. Moreover, newly added length-aware
variables are fresh and thus cannot occur in
previously processed predicates.

All remaining rules leave processed
predicates and their regular assignments unchanged,
introducing only new predicates into the frontier
and possibly fresh length-aware variables.
Therefore, (I2) is preserved.

\emph{(I3).}
Only \textsc{Subst} modifies $\SEq$.
Suppose it applies $\sigma=\{x\mapsto s\}$,
where $x\notin\Var(s)$.
By (I3), the left-hand-side variables of
existing bindings do not occur in $\mathsf N$.
Hence, $x$ is distinct from these variables,
and $s$ contains none of them.
Since $\sigma$ is applied globally to
$\mathsf N$ and $\SEq$ before adding
$x\approx s$, the eliminated variable $x$
occurs neither in the resulting predicates
nor on the right-hand side of any binding.
Thus, \textsc{Subst} preserves (I3).
All other rules leave $\SEq$ unchanged and
introduce no occurrences of eliminated
variables, so they preserve (I3) as well.

Thus, every successor state satisfies (I1)--(I3),
completing the induction.
\end{proof}

\GeneralSoundness*
\begin{proof}
By Lemma \ref{lem:state-invariants},
(I1)--(I3) hold for every reachable state.
Consider a terminal state $S$ with
$\mathsf F_S=\varnothing$.
By (I1) and (I2), $\mathsf N_S$ admits a valid
inclusion-graph representation and is stable
under $\Reg_S$. Moreover, length-aware variables
occur only on the right-hand sides of inclusions,
while (I3) ensures that $\SEq_S$ is in solved form.
The predicates in $\mathsf N_S$ impose no additional
restrictions on the non-length-aware variables
beyond their regular constraints. Consequently,
$\Phi(S)\land\Len_\varphi$ and 
$\texttt{LenImg}(\Reg_S\land\SEq_S)\land\Len_\varphi$
are equisatisfiable.

It remains to relate the terminal states to the
original input. By
Lemmas~\ref{lem:equality-aware-soundness}
and~\ref{lem:loop-reasoning-soundness},
the newly introduced equality-aware and
loop-reasoning rules preserve solutions across
their successor branches, up to existential
quantification of fresh variables.
The inherited rules preserve the same equivalence
by the solution-preservation properties of the
original stabilization procedure
\cite{blahoudek2023word,chocholaty2026string}.
Consequently, at the beginning of every iteration in Algorithm
\ref{alg:search-procedure}, the disjunction
of $\Phi(S)\land\Len_\varphi$ over all states
$S$ in the worklist $W$, with their auxiliary variables
existentially quantified, remains equivalent to
the original formula $\varphi$.
Thus, if the procedure returns $\sat$, then
$\Phi(S)\land\Len_\varphi$ is satisfiable for some
terminal state $S$, so $\varphi$ is satisfiable.
Conversely, if the procedure returns $\unsat$,
the worklist is empty. Hence, $\varphi$ is
unsatisfiable, as it is equivalent to an empty
disjunction.
\end{proof}

\end{document}